\documentclass[10pt,letterpaper,twocolumn]{article}

\usepackage[
  top=0.6in,
  bottom=0.8in,
  left=0.74in,
  right=0.74in
]{geometry}
\usepackage{enumitem}
\usepackage{libertine}
\usepackage[libertine]{newtxmath}

\usepackage{amsmath}

\usepackage{amssymb}
\usepackage{amsthm}

\allowdisplaybreaks
\newtheorem{theorem}{Theorem}

\newtheorem{corollary}{Corollary}

\theoremstyle{definition}
\newtheorem{definition}{Definition}

\theoremstyle{remark}

\usepackage{graphicx}
\usepackage{float}
\usepackage{booktabs}
\usepackage{subcaption}
\usepackage{multirow}
\usepackage{xcolor}

\usepackage{algorithm}
\usepackage{algpseudocode}

\usepackage[hidelinks]{hyperref}

\usepackage{authblk}

\let\oldbibliography\thebibliography
\renewcommand{\thebibliography}[1]{%
  \oldbibliography{#1}%
  \setlength{\itemsep}{1pt}%
  \setlength{\parskip}{1pt}%
  \setlength{\parsep}{1pt}%
}

\usepackage[numbers,sort&compress]{natbib}

\newcommand{\etal}{\textit{et al.}}

\title{Enabling Multi-Client Authorization in Dynamic SSE}

\author[1]{Seydina Ousmane Diallo}
\author[1]{Maryline Laurent}
\author[1]{Nesrine Kaaniche}

\affil[1]{SAMOVAR, Télécom SudParis, Institut Polytechnique de Paris, 91120 Palaiseau, France

\textit{\{seydina-ousmane.diallo, maryline.laurent, nesrine.kaaniche\}@telecom-sudparis.eu}}

\date{}

\begin{document}

\maketitle

\begin{abstract}
Outsourcing encrypted data to the cloud creates a fundamental tension between data privacy and functional searchability. 
Current Searchable Symmetric Encryption (SSE) solutions frequently have significant limitations, such as excessive metadata leakage, or a lack of fine-grained access control. These issues restrict the scalability of secure searches in real-world applications where multiple clients require different levels of authorization.
Our paper proposes MASSE, a dynamic multi-client SSE scheme incorporating attribute-based access control, which expands the OXT framework. With MASSE, clients are restricted sto searching for keywords authorized by their specific attribute sets, and the server remains unaware of the keywords and attributes. MASSE supports practical dynamic updates to documents, and client authorizations, including revocation, without requiring re-encryption of the database or indices, or a large number of interactions.
We formally prove the security of MASSE, that is, forward and backward privacy under a well-defined leakage profile, and token unforgeability. 
An experimental evaluation in a database containing 100 keywords, each associated with 150 documents, demonstrates the practical efficiency of MASSE. It takes less than two seconds to generate 10–100 keyword queries and 14 seconds to retrieve 50 matching documents. 
Theoretical results show that MASSE outperforms competing solutions, including OXT, and can be scaled to large encrypted databases. MASSE is also suitable for dynamic cloud deployments.

\end{abstract}
\section{Introduction}
To protect the privacy of the large volumes of data they store, stakeholders are accustomed to outsourcing it to the cloud in an encrypted form. This raises the issue of how to search for these data without decrypting the entire database and allowing other stakeholders to access it~\cite{shar}.

Searchable Encryption (SE)~\cite{article} has emerged as a promising cryptographic primitive that enables keyword-based searches to be performed directly over encrypted data. SE allows clients to query outsourced encrypted databases and retrieve relevant results without revealing either the queried keywords or the underlying data to the outsourcing server.
Since its introduction by Song \etal~\cite{SONG0}, SE has evolved into two main paradigms. \textsf{Symmetric Searchable Encryption (SSE)}~\cite{survey} enables a data owner to encrypt a collection of documents using a secret key and to perform keyword searches using the same key. In contrast, \textsf{Public Key Encryption with Keyword Search (PEKS)}~\cite{Boneh1} relies on public key cryptographic primitives to allow searches over encrypted data, offering greater flexibility at the cost of increased computational overhead.
SSE schemes are generally efficient and well-suited to large-scale cloud environments~\cite{survey}. However, most existing constructions assume a single-user setting and lack practical features such as support for dynamic updates, fine-grained access control, or efficient revocation. These limitations pose significant challenges when extending SSE to multi-client scenarios, in which multiple authorized clients must be able to search securely over an outsourced encrypted data. Without careful architectural design, malicious clients could collude to escalate their privileges beyond their authorized rights.
 Most existing SSE schemes~\cite{curt, KM10, Bost17} were originally designed for single-keyword searches, retrieving documents associated with a single keyword per query. This main functionality limits the support for more complex queries required by modern applications. To address this issue, Cash \etal~\cite{Cash14} proposed the \textsf{OXT} scheme, which constructs oblivious cross-tags to enable efficient conjunctive keyword searches with sublinear complexity.
However, \textsf{OXT} was not designed for multi-client environments and remains a static scheme. It does not support the efficient updating of data, such as the addition or deletion of encrypted documents. This limits its applicability in dynamic settings where the database evolves over time.
Several works~\cite{DU20, pat} have extended \textsf{OXT} to multi-client and dynamic scenarios, but these often compromise privacy guarantees or incur significant computational and storage overhead.

\paragraph*{Our Contributions.} 
\textsf{MASSE} extends the \textsf{OXT} framework to support multi-client scenarios with fine-grained access control and efficient revocation, while providing formally defined security guarantees.
The \textsf{MASSE} scheme relies on two key structures: one that links each keyword to a set of attributes, and another that links each attribute to its authorized keywords. These structures are used to encrypt the searchable index, which maps each keyword to the set of documents containing it. This enables clients to generate search tokens only for keywords permitted by their authorized attributes. 
The attributes themselves are never exposed in the encrypted database, thus preserving client privacy and preventing privilege escalation. 
This design means that the encrypted database only needs to be created once, with each keyword being encrypted only once. This avoids per-client re-encryption and reduces both storage and computational overheads, even in large-scale deployments. 
\textsf{MASSE} supports dynamic operations, enabling the data owner to efficiently insert or delete documents, register new clients with additional attribute sets, and manage user revocation. 
The revocation mechanism is lightweight and immediately prevents revoked clients from performing searches.
Finally, we conduct a formal security analysis based on clearly defined leakage profiles and evaluate performance using a prototype implementation.

The remainder of the paper is organized as follows. Section~\ref{sec: rel} surveys the related work. Section~\ref{sec: system} introduces our system and threat models. Section~\ref{sec: Pre} provides an overview of our proposed \textsf{MASSE} scheme. Section~\ref{sec: conc} gives the concrete construction of \textsf{MASSE}, with detailed algorithms. Section~\ref{sec: security-analysis} presents a security analysis of \textsf{MASSE}, while Section~\ref{sec: Dis} provides a performance analysis. Section~\ref{sec:con} concludes the paper. 

\section{Related Work}
\label{sec: rel}

Song \etal~introduced the first SSE scheme enabling keyword search over encrypted data~\cite{SONG0}. 
Its main limitation is its inefficiency, since each keyword must be individually encrypted and exhaustively checked during the search. 
Subsequent works improved efficiency, strengthened security, and added the following functionalities, as summarized in Table~\ref{tab:comparison}. 

-- \textsf{Query expressiveness:} 
Query expressiveness refers to the ability of a searchable encryption scheme to support queries that are more complex than a single-keyword search, such as conjunctive keyword queries.
Building on Song \etal's proposal, Goh~\cite{Goh} introduced Bloom filter–based indices combined with pseudorandom functions (PRFs) to improve search efficiency and achieve semantic security against chosen-keyword attacks. 
However, Goh’s scheme still required a linear search time that was proportional to the number of documents. 
Curtmola \etal~\cite{curt} addressed this limitation by using inverted indices, to achieve sub-linear search time, and they formalized SSE security within the real/ideal world framework. 
However, these early works, focused on single-keyword searches, which limits their applicability to large-scale databases.
 
Several schemes have been proposed to support more expressive queries and address conjunctive searches. Golle \etal~\cite{Gol} proposed a framework for keyword conjunctions, which was later extended by Ballard \etal~\cite{Ba} and Jin \etal~\cite{Jin} for multi-keyword operations. 
While these constructions allowed conjunctive queries, the search cost remained linear in the database size. 
A major improvement was introduced by Cash \etal~\cite{Cash14} in their work on \textsf{OXT}, the first practical searchable encryption scheme achieving sub-linear search time for conjunctive queries. \textsf{OXT} employs a secure inverted index, called \(\mathsf{TSet}\), to retrieve candidate documents matching the least frequent keyword, and an algebraic verification structure, called \(\mathsf{XSet}\), to efficiently verify the remaining keywords. However, \textsf{OXT} inherently leaks correlations between queried keywords, specifically which documents contain multiple keywords together, a leakage formally known as the keyword-pair result pattern (KPRP). To address this issue, Lai \etal~\cite{Lai} proposed \textsf{HXT}, which extends \textsf{OXT} by obfuscating the verification of secondary keywords through query-dependent randomization, thereby preventing KPRP leakage.

-- \textsf{Dynamic Updates:} 
Despite these advances, the above schemes remain static and unable to handle insertions or deletions efficiently. 
Kamara \etal~\cite{KMR12} proposed one of the first dynamic SSE (DSSE) schemes, which uses keyword-indexed linked lists and auxiliary encrypted structures to support updates while maintaining efficient searching. However, access patterns were still leaked. Stefanov \etal~\cite{Stf2013} introduced forward privacy (FP) and backward privacy (BP): FP ensures that new insertions do not reveal links to previous queries, and BP ensures that deletions do not affect future queries. 
Bost \etal~\cite{Bost16,Bost17} formalized FP and BP, and categorized three levels of backward privacy. Type-I BP, the strongest, reveals only documents matching the current query. Type-II BP additionally exposes timestamps, i.e., the times at which insertions or deletions occur. Type-III BP further discloses which deletions cancel prior insertions, allowing the server to potentially determine the list of added or deleted documents and correlate operations over time.

Using these definitions, Bost \etal~\cite{Bost17} constructed DSSE schemes with different privacy levels. MONETA (Type-I) uses a Timestamped Write-Authenticated Memory (TWRAM) along with PRFs and multi-set hash functions to authenticate query results. FIDES achieves Type-II backward privacy with moderate efficiency. JANUS and DIANA$_{del}$ are Type-III schemes optimized for performance.
Sun \etal~\cite{Sun} further improved JANUS with Janus++, using Symmetric Puncturable Encryption (SPE) to efficiently maintain Type-II backward privacy. 
ODXT, recently proposed by Patranabis \etal~\cite{pat}, extends DSSE to support conjunctive queries with forward and Type-II backward privacy, but it does not support multiple clients.

-- \textsf{Multi-client SSE (MSSE):} 
Building on dynamic SSE schemes for single-keyword and conjunctive queries, multi-client support with access control and revocation has been explored. However, careful design is still required to balance security and efficiency.
Early schemes, such as that by Curtmola \etal~\cite{curt}, required data to be duplicated for each client, resulting in significant storage overhead on both the server and clients. 

Sun \etal~\cite{sun16} applied attribute-based encryption to enforce access control in a non-interactive manner, enabling the server to verify users' permissions without requiring direct interaction with the data owner. However, dynamic updates were limited, and storage efficiency remained an issue.
 
Du \etal~\cite{DU20} proposed a dynamic MSSE scheme based on \textsf{OXT}, in which clients are granted access to keyword ranges certified by digital signatures issued by the data owner. New authorizations can be added without involving existing clients. In Du \etal~\cite{DU20}'s scheme, revocation requires updating the index entries corresponding to removed clients. However, the scheme does not support dynamic updates of the underlying database, and sharing keywords among multiple clients leads to duplicated search indices, thereby increasing storage overhead.
Wanshan \etal~\cite{X2022DynamicSSE} proposed the \textsf{DMC-SSE} scheme, which combines searchable symmetric encryption with access control lists (ACLs) to support multi-client environments and boolean queries. Although this design simplifies user addition and revocation, the size of the ACL increases linearly with the number of users.
Building on ODXT~\cite{pat}, Hu \etal~\cite{Hu2025SEAC} proposed \textsf{SEAC}, a dynamic multi-client SSE scheme using a dynamic RSA-based accumulator for access control. Client permissions for search and updates are verified via accumulator membership proofs maintained by the data owner.
In \textsf{SEAC}, updates and permission management incur significant overhead: each update involves interaction among the client, data owner, and server, along with multiple exponentiations to maintain the accumulator. Client authorization and revocation require updating the accumulator and broadcasting auxiliary information, forcing all remaining clients to locally update their witnesses using the extended Euclidean algorithm.

Finally, \textsf{SEAC} incorporates a random blinding factor to prevent linking multiple requests across sessions.

-- \textsf{Client Administration:} 
This feature measures the practicality of implementing the solution.  
In contrast to \textsf{MASSE}, \textsf{SEAC} is deemed impractical, since updating or revoking a client require interactions with each of the other clients. 


\begin{table}[h]
  \caption{Comparison of Representative SSE Schemes}
  \label{tab:comparison}
  \centering
  \resizebox{\columnwidth}{!}{%
  \begin{tabular}{lcccccc} 
    \toprule
    Scheme & Query & MC & Dyn & FP & BP & CA \\
    \midrule
    Song \etal~\cite{SONG0} & single & $\times$ & $\times$ & $\times$ & $\times$ & $\times$ \\
    Kamara \etal~\cite{KMR12} & single & $\times$ & \textcolor{blue}{$\checkmark$} & $\times$ & $\times$ & $\times$ \\
    Cash \etal~(\textsf{OXT})~\cite{Cash14} & conj/bool & $\times$ & $\times$ & $\times$ & $\times$ & $\times$ \\
    Bost \etal~\cite{Bost16} & single & $\times$ & \textcolor{blue}{$\checkmark$} & \textcolor{blue}{$\checkmark$} & \textcolor{blue}{$\checkmark$} & $\times$ \\
    Patranabis \etal~\cite{pat} & conj/bool & $\times$ & \textcolor{blue}{$\checkmark$} & \textcolor{blue}{$\checkmark$} & $\times$ & $\times$ \\
    Curtmola \etal~\cite{curt} & single & \textcolor{blue}{$\checkmark$} & $\times$ & $\times$ & $\times$ & $\times$ \\
    Du \etal~(\textsf{DMSSE})~\cite{DU20} & conj/bool & \textcolor{blue}{$\checkmark$} & \textcolor{blue}{$\checkmark^*$} & $\times$ & $\times$ & $\times$  \\
    Hu \etal~(\textsf{SEAC})~\cite{Hu2025SEAC} & conj/bool & \textcolor{blue}{$\checkmark$} & \textcolor{blue}{$\checkmark$} & \textcolor{blue}{$\checkmark$} & \textcolor{blue}{$\checkmark$} & $\times$\\
    \textbf{\textsf{MASSE}} (ours) & conj/bool & \textcolor{blue}{$\checkmark$} & \textcolor{blue}{$\checkmark$} & \textcolor{blue}{$\checkmark$} & \textcolor{blue}{$\checkmark$} & \textcolor{blue}{$\checkmark$} \\
    \bottomrule
  \end{tabular}%
  }
  \footnotesize{
    MC: Multi-Client support; Dyn: Dynamic updates; FP: Forward Privacy; BP: Backward Privacy; CA: Client Administration. \\
    \textcolor{blue}{$\checkmark^*$} indicates partial dynamic support (updates apply only to client permissions).
  }
\end{table}
\section{System and Threat Models}
\label{sec: system}
This section presents the \textsf{MASSE} scheme. We first define the notations, then describe the system model, followed by the threat model, and finally summarize the security and privacy properties.

\subsection{Notations}
Table~\ref{tab:notations} introduces the main notations and abbreviations used throughout the paper. 
Let $\lambda \in \mathbb{N}$ be the security parameter. 
The set $\{0,1\}^\lambda$ denotes all binary strings of length $\lambda$, and $\{0,1\}^*$ represents binary strings of arbitrary length. 
The cardinality of a vector $v$ is written $|v|$, and its $i$-th component is denoted $v[i]$ for $i \in [|v|]$. 
For any positive integer $n$, we define $[n] = \{1,2,\dots,n\}$. 
Let $d$ denote the number of documents and $D$ the number of distinct keywords in the database. 
If $A$ is a deterministic algorithm, we write $y \gets A(x)$ for its output on input $x$. 
For a probabilistic algorithm, the output is a random variable, denoted by $y \overset{\$}{\gets} A(x)$. A conjunctive query is denoted by $q = \Psi(\overline{W})$, where $\Psi$ is a boolean formula over a subset of keywords $\overline{W} \subseteq W$.

\begin{table}[h]
\small
\centering
\caption{Summary of notations used in our scheme}
\label{tab:notations}
\renewcommand{\arraystretch}{1.1}
\resizebox{\columnwidth}{!}{%
\begin{tabular}{cl}
\toprule
\textbf{Symbol} & \textbf{Description} \\
\midrule
$\mathsf{DB}$ & Plaintext database: $\{(\mathsf{ind}_i, w_j)\}^{d,D}_{i=1,j=1}$ \\
$\mathsf{EDB}$ & Encrypted database \\
$i$ & Index of the $i$-th document \\
$j$ & Index of the $j$-th keyword \\
$t$ & Index of the $t$-th attribute \\
$c$ & Index of the client \\
$w_j$ & $j$-th keyword associated with one or more documents \\
$\text{ind}_i$ & Unique identifier of the $i$-th document \\
$W$ & Set of all distinct keywords \\
$\overline{W}_\mathcal{C}$ & Keywords authorized for client $\mathcal{C}$ \\
$q = \Psi(\overline{W}_\mathcal{C})$ & Boolean query over keywords in $\overline{W}_\mathcal{C}$ \\
$\mathsf{DB}(q)$ & Documents satisfying query $q$ \\
$\mathsf{DB}(w_j)$ & Documents containing keyword $w_j$ \\
$\mathsf{A}$ & Set of attributes $a$ \\
$\mathsf{A}_\mathcal{C}$ & Attributes associated with client $\mathcal{C}$ \\
$W_t$ & Keywords associated with attribute $a_t$ \\
$A_j$ & Attributes associated with keyword $w_j$ \\
$\mathsf{S}_x = \{(w_j, x)\}$ & Keyword–key pairs associated with secret $x$ \\
$K$ & Master secret key of the scheme \\
$v_j$ & Encryption key of document identifiers linked to $w_j$ \\
$h_j$ & Key derived from the relation between $a_t$ and $w_j$ \\
$\textsf{Tk}_q$ & Search token generated for query $q$ \\
$\mathsf{R}$ & Search result set returned by server $\mathcal{P}$ \\
$\mathsf{keys}_\mathcal{C}$ & Keys assigned to client $\mathcal{C}$ \\
$\mathsf{D}_\mathcal{C}$ & Access dictionary entry for client $\mathcal{C}$ \\
\bottomrule
\end{tabular}%
}
\end{table}

\subsection{System Model } 
\textsf{MASSE} relies on three main entities: the Data Owner $\mathcal{O}$, the Cloud Server $\mathcal{P}$ (for Provider), and the Client $\mathcal{C}$.

\begin{itemize}
    \item \textsf{Data Owner} $\mathcal{O}$. It initializes the system by encrypting the database $\mathsf{EDB}$ and constructing search indices. 
    It associates each client's attribute (or a set of attributes) with a subset of authorized keywords $\overline{W}_\mathcal{C}$ and generates dedicated search keys for each client $\mathcal{C}$. 
    $\mathcal{O}$ signs the ranges of authorized keywords per client to create authorization tokens, and sends these signatures along with the encrypted database $\mathsf{EDB}$ and search indices to $\mathcal{P}$ to enforce access control. 
    It is also responsible for performing dynamic updates on the database and client authorizations.
    
    \item \textsf{Cloud \, Server \quad $\mathcal{P}$}. The encrypted database $\mathsf{EDB}$ and its associated search indices are stored by $\mathcal{P}$. Upon receiving a search request from a client, $\mathcal{P}$ first verifies the client's authorization using the corresponding authorization token. $\mathcal{P}$ then runs the search algorithm and returns the matching results in encrypted form. Furthermore, the  server handles update operations, such as adding or deleting document-keyword pairs, when receiving update instructions from the data owner.

    \item \textsf{Client} $\mathcal{C}$. This is an authorized client  performing encrypted search queries on the $\mathsf{EDB}$. 
    $\mathcal{C}$ holds secret keys and an authorization token issued by $\mathcal{O}$. It generates secure search queries and sends them to $\mathcal{P}$. $\mathcal{C}$ is only permitted to query the subset of keywords for which it is authorized, and it only receives the encrypted identifiers of the matching documents. 
    \end{itemize}



\subsection{Threat Model}
In our threat model, $\mathcal{O}$ is fully trusted.
$\mathcal{C}$ may behave maliciously by generating unauthorized search tokens or colluding with other clients.
$\mathcal{P}$ is honest-but-curious, meaning that it follows the protocol but attempts to infer information from encrypted data or queries.
We assume that any collusion is limited to clients, and that the server does not collaborate with them.

\subsection{Security and Privacy Properties}
\label{subsec:security-properties}
\textsf{MASSE} should fulfill the following security properties:

\textbf{Indistinguishability of Data and Queries.}
Assuming the server to be honest-but-curious, it should not learn any information about the encrypted database or the issued queries beyond what is explicitly revealed by the leakage function.  

\textbf{Token Unforgeability.} 
Each client is restricted to its authorized subset of keywords. 
Even if they collude, it is computationally infeasible for malicious clients to generate a valid search token for any keyword they are not authorized to use. 
This prevents privilege escalation and enforces strict access control.

\textbf{Forward Privacy.}
Forward privacy ensures that the addition of new documents does not reveal information about previous search queries or indexed data.

\textbf{Backward Privacy.}
\textsf{MASSE} guarantees Type-II backward privacy: upon a deletion operation, the server may learn that a deletion has occurred, but cannot identify the deleted document nor link it to any previous or subsequent search queries.

\section{Overview of MASSE}
\label{sec: Pre}
This section provides a brief overview of the \textsf{MASSE} scheme, highlighting its foundation on \textsf{OXT} and its extensions for multi-client authorization and dynamic operations.

\subsection{Introduction to OXT Scheme}
In line with the standard SSE syntax, we provide an overview of the \textsf{OXT} protocol~\cite{Cash14}, focusing on the Setup and Search phases. The \textsf{OXT} scheme enables sub-linear, privacy-preserving conjunctive keyword searches through two encrypted data structures, namely $\mathsf{Tset}$ and $\mathsf{Xset}$. It is defined by the tuple
\(\pi_{\textsf{OXT}} = (\mathsf{Setup}, \mathsf{Search}).\)

The structure $\mathsf{Tset}$ consists of secure inverted lists: for each keyword $w$, it stores encrypted identifiers of documents containing $w$. The structure $\mathsf{Xset}$ is an algebraic verification set containing elements derived from pairs $(w, ind)$, enabling conjunctive checks without revealing additional information.

\paragraph*{Setup Phase.}
\((K, \mathsf{EDB}) \leftarrow \mathsf{Setup}(\lambda, \mathsf{DB})\)
Given a security parameter $\lambda$ and a keyword-indexed database $\mathsf{DB}$, the owner $\mathcal{O}$ proceeds as follows:

\begin{enumerate}
    \item $\mathcal{O}$ generates secret keys 
    $K = (k_S, k_I, k_T, k_z, k_X)$, each sampled uniformly from $\{0,1\}^\lambda$. 
    Let 
    \(
    F: \{0,1\}^\lambda \times \{0,1\}^* \rightarrow \{0,1\}^\lambda
    \quad \text{and} \quad
    F_p: \{0,1\}^\lambda \times \{0,1\}^* \rightarrow \mathbb{Z}_p^*
    \)
    be pseudo-random functions.

    \item For each keyword $w \in W$ and document identifier $ind \in DB(w)$, and for each counter $c$ associated with occurrences of $w$, $\mathcal{O}$ computes
    \(
    k_e = F(k_S, w), \quad
    \mathsf{xind} = F_p(k_I, ind), \quad
    z = F_p(k_z, w \,\|\, c),
    \)
    then $\mathcal{O}$ sets
    \(y = \mathsf{xind} \cdot z^{-1}.\)
    The identifier is encrypted as
    \(e_2 = \mathsf{Sym.Enc}(k_e, ind).\)
    Next, $\mathcal{O}$ computes the search tag
    \(\mathsf{stag} = F(k_T, w)\)
    and stores the pair $(y, e_2)$ in $\mathsf{Tset}[\mathsf{stag}]$. It also computes the verification tag
    \(\mathsf{xtag} = g^{F_p(k_X, w)\cdot \mathsf{xind}}\)
    and inserts it into $\mathsf{Xset}$.
    Finally, the encrypted database is defined as
    \(\mathsf{EDB} = (\mathsf{Tset}, \mathsf{Xset}).\)
\end{enumerate}

\paragraph*{Search Phase.}
\(\mathsf{R} \leftarrow \mathsf{Search}(K, q, \mathsf{EDB})\)
To evaluate a conjunctive query $q = \psi(w_1, \dots, w_n)$, $\mathcal{O}$ selects the least frequent keyword $w_1$ as the pivot (s-term), while the remaining keywords $w_2, \dots, w_n$ serve as x-terms.
First, $\mathcal{O}$ computes
\(\mathsf{stag} = F(k_T, w_1).\)
Then, for each keyword $w_j$ with $j \ge 2$ and for each counter $c$ associated with $w_1$, it computes
\(\mathsf{xtoken}[c, j] = g^{F_p(k_z, w_1 \,\|\, c)\cdot F_p(k_X, w_j)}.\)

Upon receiving $\mathsf{stag}$ and the corresponding $\mathsf{xtoken}[c,j]$, the server $\mathcal{P}$ retrieves all pairs $(y, e_2)$ from $\mathsf{Tset}[\mathsf{stag}]$. For each entry and each $j \ge 2$, it verifies whether
\(\mathsf{xtoken}[c,j]^y \in \mathsf{Xset}.\)
If all checks are successful, the ciphertext $e_2$ is added to the set of results $\mathsf{R}$. Finally, $\mathcal{O}$ decrypts each $e_2$ using $k_e$ to recover the identifiers of the matching document.

\subsection{MASSE as an Extension of \textsf{OXT} Scheme}
We introduce \textsf{MASSE}, an SSE scheme that extends \textsf{OXT} to a dynamic multi-client setting with attribute-based access control. In \textsf{MASSE}, the search index is encrypted using keys derived from keyword–attribute relationships, enabling fine-grained authorization without relying on public-key encryption.
\textsf{MASSE} adopts $\mathsf{Tset}$ and $\mathsf{Xset}$ from \textsf{OXT} and further introduces a new structure $\mathsf{Cset}$ based on a single tag $\mathsf{ctag}$ per keyword. In contrast to the scheme of Du \etal~\cite{DU20}, where one $\mathsf{ctag}$ is generated for each keyword–client pair, our construction avoids such duplication. The $\mathsf{ctag}$ serves as a cryptographic witness that authenticates authorized clients during search.
Client permissions are managed through a dynamic authorization dictionary, which supports efficient client addition, revocation, and enforcement of attribute-based access policies. Figure~\ref{fig:sequence_diagram} illustrates the core interactions of the system, where each entity independently executes $\textsf{KeyGen}$ to generate its own key pair.

\begin{figure*}[h]
  \centering
  \includegraphics[width=0.68\textwidth, keepaspectratio]{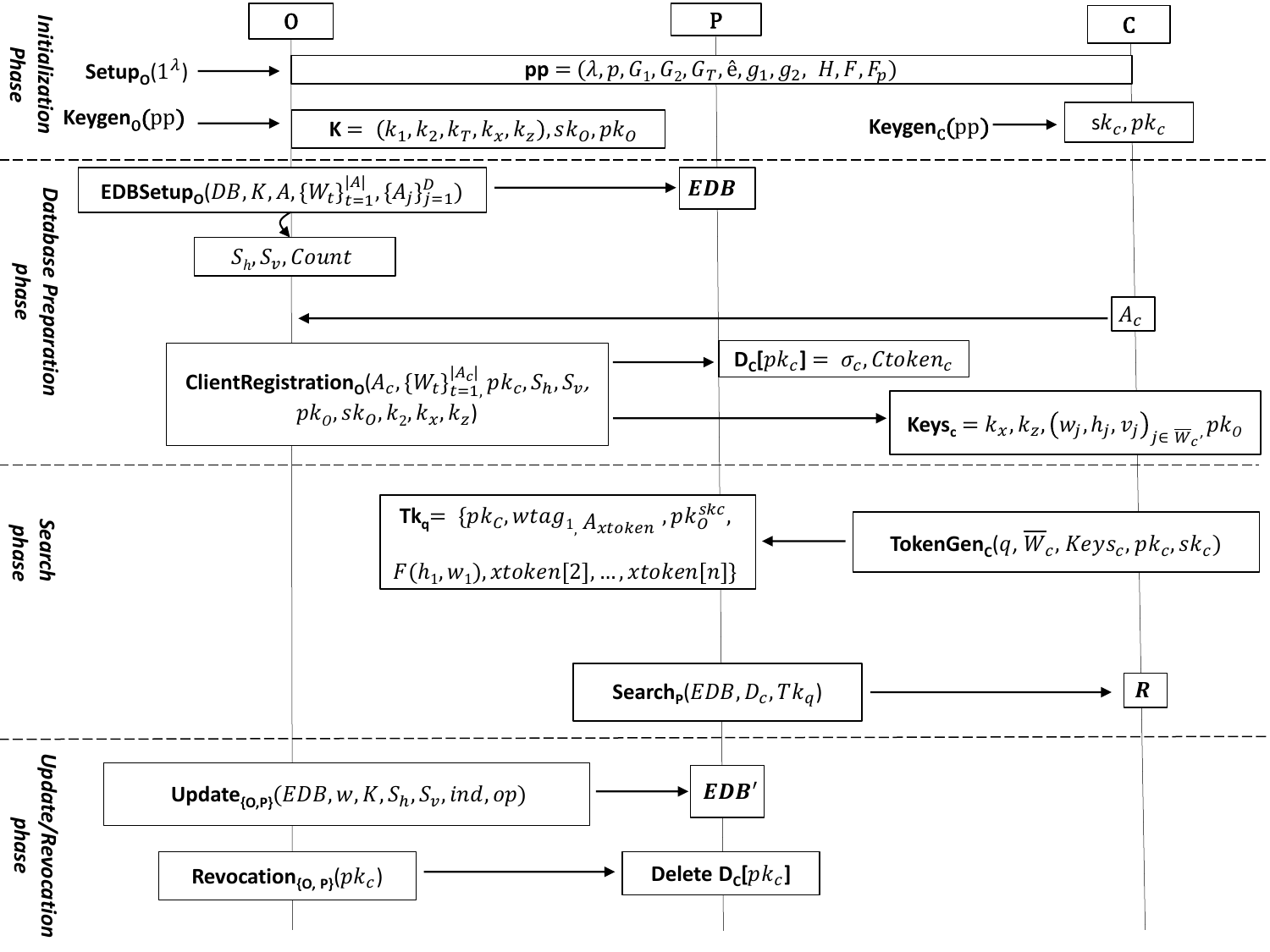}
  \caption{High-level system workflow of the MASSE scheme}
  \label{fig:sequence_diagram}
\end{figure*}

\subsection{MASSE Building Blocks}
\label{subsec:phases}

The \textsf{MASSE} scheme is organized into four main phases: \textsf{Initialization}, \textsf{Database Preparation}, \textsf{Search}, and \textsf{Update/Revocation}. 
Each phase consists of several algorithms, as described below.

\textsf{Initialization.} This phase sets up the foundational cryptographic environment for the \textsf{MASSE} system. 
It generates the public parameters and the key material required by both $\mathcal{O}$ and the clients, and includes the following algorithms:
\begin{itemize}
    \item $\mathsf{pp} \gets \mathsf{Setup}_\mathcal{O}(1^\lambda)$. Run by $\mathcal{O}$. 
    It takes the security parameter $1^\lambda$ as input and outputs the public parameters $\mathsf{pp}$.
    
    \item $(K,\,(sk_\mathcal{O}, pk_\mathcal{O})) \gets \mathsf{KeyGen}_\mathcal{O}(\mathsf{pp})$. 
    Run by $\mathcal{O}$, it outputs the master secret key $K$ and the owner secret/public key pair $(sk_\mathcal{O} \in \mathbb{Z}_p, pk_\mathcal{O} = g_2^{sk_\mathcal{O}})$.
    
    \item $(sk_\mathcal{C}, pk_\mathcal{C}) \gets \mathsf{KeyGen}_\mathcal{C}(\mathsf{pp})$. 
    Run by $\mathcal{C}$, it outputs the client secret/public key pair $(sk_\mathcal{C} \in \mathbb{Z}_p, pk_\mathcal{C} = g_1^{sk_\mathcal{C}})$.
\end{itemize}

\textsf{Database Preparation.} 
This phase involves preparing the encrypted database and distributing client credentials. 
It ensures that each client can only search for keywords authorized by their attributes and includes the following algorithms:

\begin{itemize}
    \item $(\mathsf{EDB},\mathsf{S}_h,\mathsf{S}_v, \mathsf{Count}) \gets \mathsf{EDBSetup}_\mathcal{O}\big( \mathsf{DB}, K, A, \\ \{\mathsf{W}_t\}_{t=1}^{|A|}, \{\mathsf{A}_j\}_{j=1}^{D} \big)$. 
    Run by $\mathcal{O}$, this algorithm encrypts the database and initializes:
    a structured encrypted database $\mathsf{EDB}$,
    a table of key-attribute keys $\mathsf{S}_h$, 
    a table of document keys $\mathsf{S}_v$ and a table of counter
    $\mathsf{Count}$, which stores the number of documents associated with each keyword.
    For each keyword, a fixed number of encrypted dummy documents are added during the database encryption, in order to allow future updates.
    \item $(\mathsf{D}_\mathcal{C},\, \mathsf{keys}_\mathcal{C}) \gets 
    \mathsf{ClientRegistration}_\mathcal{O}\big(
        \mathsf{A}_\mathcal{C},
        \{\mathsf{W}_t\}_{t=1}^{|\mathsf{A}_\mathcal{C}|}, \\
        pk_\mathcal{C}, \mathsf{S}_h, \mathsf{S}_v 
    \big)$. 
    Run by $\mathcal{O}$, it outputs the client dictionary $\mathsf{D}_\mathcal{C}$ and the secret keys $\mathsf{keys}_\mathcal{C}$ for each client $c$.
\end{itemize}

\textsf{Search.} 
This phase allows clients to query the encrypted database while preserving the confidentiality of their queries and authorized keywords. 
It consists of the client generating a search token and the server executing the search. It includes the following algorithms:

\begin{itemize}
    \item $\textsf{Tk}_q \gets \mathsf{TokenGen}_\mathcal{C}(q, \overline{W}_\mathcal{C}, \mathsf{keys}_\mathcal{C}, pk_\mathcal{C}, sk_\mathcal{C})$. 
    Run by client $\mathcal{C}$, this algorithm produces the search token $\textsf{Tk}_q$ for the query $q \subseteq \overline{W}_\mathcal{C}$, where $\overline{W}_\mathcal{C}$ denotes the set of keywords authorized for client $c$.

    \item \(\textsf{R} \gets \mathsf{Search}_\mathcal{P}\left(\mathsf{EDB}, \mathsf{D}_\mathcal{C}, \textsf{Tk}_q\right)\). 
    Run by provider $\mathcal{P}$, this algorithm first verifies the client's authorization. 
    If the client is authorized, it returns the encrypted set of document identifiers matching the query, denoted by $\mathsf{R}$. 
    Otherwise, it returns an empty set.

\end{itemize}

\textsf{Update and Revocation.} 
This phase enables dynamic operations on the database and enforces client revocation. It includes:

\begin{itemize}
  \item $\mathsf{EDB'} \gets \mathsf{UpdateEDB}_{\mathcal{O}, \mathcal{P}}(\mathsf{EDB}, w_j, K, \mathsf{S}_h, \mathsf{S}_v, ind_i, op)$. 
  Jointly run by $\mathcal{O}$ and $\mathcal{P}$, it updates the encrypted database $\mathsf{EDB}$ according to insertion or deletion operations, producing an updated version $\mathsf{EDB'}$ with refreshed indices.

  \item $\mathsf{Revoke}_{\mathcal{O}, \mathcal{P}}(pk_\mathcal{C}, \mathsf{D}_\mathcal{C})$.
  Jointly run by $\mathcal{O}$ and $\mathcal{P}$, this operation takes as input the client public key $pk_\mathcal{C}$ and the corresponding authorization dictionary $\mathsf{D}_\mathcal{C}$, and revokes the client’s authorization without affecting other clients or re-encrypting the database.
\end{itemize}

\section{Concrete Construction of MASSE}
\label{sec: conc}
This section provides a detailed description of the concrete construction of \textsf{MASSE}, including all the algorithms and phases.
\subsection{ \textsf{MASSE} Algorithms}
The \textsf{MASSE} construction is formally defined by the following suite of algorithms: the foundational procedures $\mathsf{Setup}_{\mathcal{O}}$,$\mathsf{KeyGen}_{\mathcal{O}}$, and $\mathsf{KeyGen}_{\mathcal{C}}$, and followed by the operational algorithms $\mathsf{EDBSetup}_{\mathcal{O}}$, $\mathsf{Registration}_{\mathcal{O, C}}$, $\mathsf{TokenGen}_{\mathcal{C}}$, $\mathsf{Search}_{\mathcal{P}}$, $\mathsf{Update}_{\mathcal{O, P}}$, and $\mathsf{Revocation}_{\mathcal{O, P}}$.

\begin{itemize}
    \item $\mathsf{pp} \leftarrow \mathsf{Setup}_\mathcal{O}(1^\lambda)$:
    given a security parameter $\lambda$, $\mathcal{O}$ generates the public parameters
    \(\mathsf{pp} \footnote{All entities are assumed to use the same system parameters pp unless otherwise stated.}\), where: \\
    \(\mathsf{pp} = (\lambda, p,\mathbb{G}_1,\mathbb{G}_2, \mathbb{G}_T, \hat{e}, g_1, g_2, H, F, F_p)\), 
    with $p$ being a prime number of $\lambda$ bits and 
    $\mathbb{G}_1$, $\mathbb{G}_2$, and $\mathbb{G}_T$ are cyclic groups of order $p$ and $\hat{e}: \mathbb{G}_1 \times \mathbb{G}_2 \rightarrow \mathbb{G}_T$ is a bilinear pairing. $g_1$ and $g_2$ are generators of $\mathbb{G}_1$ and $\mathbb{G}_2$, respectively. $H: \{0,1\}^* \rightarrow \{0,1\}^\lambda$ is a cryptographic hash function. $F: \{0,1\}^\lambda \times \{0,1\}^* \rightarrow \{0,1\}^\lambda$ is a pseudo-random function (PRF).  $F_p: \{0,1\}^\lambda \times \{0,1\}^* \rightarrow \mathbb{Z}_p^*$ is a PRF that outputs elements in the multiplicative group of the field $\mathbb{Z}_p$.
\end{itemize}
The data owner and each client generate their key pairs as follows:
    \begin{itemize}
        \item $(K, sk_\mathcal{O}, pk_\mathcal{O}) \leftarrow \mathsf{KeyGen}_{\mathcal{O}}(\mathsf{pp})$:  
        $\mathcal{O}$ samples five independent $\lambda$-bit secret keys $K = (k_1, k_2, k_T, k_x, k_z) \in \{0,1\}^\lambda$, used for encryption and secure query processing.  
        The owner also chooses a secret key $sk_\mathcal{O} \xleftarrow{\$} \mathbb{Z}_p^*$ and computes the corresponding public key $pk_\mathcal{O} = g_2^{sk_\mathcal{O}} \in \mathbb{G}_2$.

        \item $(sk_\mathcal{C}, pk_\mathcal{C}) \leftarrow \mathsf{KeyGen}_{\mathcal{C}}(\mathsf{pp})$:  
        Each client $\mathcal{C}$ independently selects a secret key $sk_\mathcal{C} \xleftarrow{\$} \mathbb{Z}_p^*$ and computes the public key $pk_\mathcal{C} = g_1^{sk_\mathcal{C}} \in \mathbb{G}_1$.
    \end{itemize}

 \paragraph{$(\mathsf{EDB}, \mathsf{S}_h, \mathsf{S}_v, \mathsf{Count}) \leftarrow 
\mathsf{EDBSetup}_\mathcal{O}(\mathsf{DB}, K, A, \{\mathsf{W}_t\}_{t=1}^{|\mathsf{A}|}, \\ \{\mathsf{A}_j\}_{j=1}^{D})$.}
Algorithm~\ref{alg:edbsetup}, run by $\mathcal{O}$, takes as input a database
\(\mathsf{DB} = \{(ind_i, w_j)\}\), master keys 
\(K = \{k_1, k_2, k_T, k_x, k_z\}\), an attribute set $\mathsf{A}$, and two sets of tables: $\{\mathsf{W}_t\}_{t=1}^{|\mathsf{A}|}$, where each $\mathsf{W}_t$ lists the keywords for the $t$-th attribute,  and $\{\mathsf{A}_j\}_{j=1}^{D}$, where each $\mathsf{A}_j$ lists the attributes associated with the $j$-th keyword.  
Here, $D$ denotes the total number of keywords in the database (DB).
The keys $(k_x,k_z)$ are later shared with clients to enable them to generate 
search queries independently of~$\mathcal{O}$. 
This algorithm outputs the encrypted database 
$\mathsf{EDB}=\{\mathsf{Tset},\mathsf{Xset},\mathsf{Cset}\}$ together with 
two key tables, 
$\mathsf{S}_h=\{(w_j,h_j)\}$ and $\mathsf{S}_v=\{(w_j,v_j)\}$, and an array 
$\mathsf{Count},$ storing the number of (real and dummy) entries per keyword. 

For each keyword $w_j$, $\mathcal{O}$ computes:
\(l_j = F(k_1, \bigoplus_{a_t\in\mathsf{A}_j} a_t),\quad
h_j = F(k_1,l_j),\quad
v_j = F(k_1,h_j),\)
where $F$ is a PRF. The key $v_j$ encrypts the search index associated with keyword $w_j$.

Each keyword receives $\alpha$ encrypted dummy entries to hide its real document frequency.

\begin{itemize}
    \item $\mathsf{Tset}\, \mathsf{T}:$ Following the \textsf{OXT} framework, the inverted index $\mathsf{T}$ links each keyword to the identifiers of encrypted documents. Each keyword $w_j$ is transformed into a search token 
    \(\mathsf{stag}_j \leftarrow F(k_T, w_j)\), where $F$ is a PRF keyed by $k_T$. 
    The search token is then indexed as
    \(l_j \leftarrow H(\mathsf{stag}_j \Vert c), \quad c \in [1, |\mathsf{DB}(w_j)|+\alpha]\), where $|\mathsf{DB}(w_j)|$ is the number of documents containing $w_j$, and $\alpha$ denotes additional dummy entries.    
    Each entry stored in $\mathsf{T}[l_j]$ is a pair $(e_{i,j}, y_{i,j})$, computed as:
    \(
    e_{i,j} = \textsf{Sym.Enc}(v_j, ind_i) \quad \text{and} \quad
    y_{i,j} = \mathsf{x}_i \cdot z_j^{-1},
    \)
    where
    \(\mathsf{x}_i = F_p(k_2, ind_i) \quad \text{and} \quad z_j = F_p(k_z, w_j).\) Additionally, for each keyword \(w_j\), \(\alpha\) encrypted dummy entries are inserted into $\mathsf{Tset}$ to simplify future updates while preserving the functional correctness of the scheme.

  \item $\mathsf{Xset} \, \mathsf{X:}$ Following the construction of \textsf{OXT}~\cite{Cash14}, this structure stores cryptographic hashes of $(w_j, ind_i)$ pairs to enable efficient conjunctive queries. Each value, denoted by $\mathsf{xtag}_{i,j} \leftarrow g_1^{F_p(k_x, \, w_j) \cdot F_p(k_2, \, ind_i)} \in \mathbb{G}_1$, allows the server to verify, for any document retrieved, whether it contains the other keywords of the query, without revealing the content of the document.
    \item $\mathsf{Cset} \,\mathsf{C:}$ This structure enables the server to access the corresponding entries in $\mathsf{Tset}$ without revealing $k_T$ to the client. 
    Each entry is indexed by a witness \(\mathsf{wtag}_j = g_1^{F_p(v_j, w_j)} \in \mathbb{G}_1\), which serves as a cryptographic proof for the authorization of the query. 
    For storage, the witness is hashed: \(\mathsf{ctag}_j = H(\mathsf{wtag}_j),\)
    which is the value stored in $\mathsf{Cset}$. $\mathcal{P}$ uses
    \(\theta_j = \mathsf{stag}_j \oplus F(h_j, w_j)\), and \(\mathsf{C}[\mathsf{ctag}_j] = \theta_j,\)
    to retrieve the search token $\mathsf{stag}_j$ and access the corresponding entries in $\mathsf{Tset}$.   
\end{itemize}

\begin{algorithm}[t]
\caption{\(\mathsf{EDBSetup}_{\mathcal{O}}\)}
\label{alg:edbsetup}
\begin{algorithmic}[1]
  \Statex \textsf{input:}
 \(\mathsf{DB}\),  \(K = (k_1, k_2, k_T, k_x, k_z)\), \(\mathsf{A}\),  \(\{\mathsf{W}_t\}_{t=1}^{|\mathsf{A}|}\), \(\{\mathsf{A}_j\}_{j=1}^D\)
  \Statex \textsf{output:}
\(\mathsf{EDB} = (\mathsf{T}, \mathsf{X}, \mathsf{C})\), \(\mathsf{S}_h\), \(\mathsf{S}_v\),  \(\mathsf{Count}\)
  \State \(\mathsf{T} \gets \emptyset,\ \mathsf{C} \gets \emptyset,\ \mathsf{X} \gets \emptyset,\, \ \mathsf{S}_h \gets \emptyset, \, \mathsf{S}_v \gets \emptyset, \, \mathsf{Count} \gets \emptyset \) 
\For{each \(w_j \in \mathsf{W}_t\)}
     \State \( l_{j,t} \gets F(k_1, \bigoplus_{t \in \mathsf{A}_j} a_t) \) \State $h_j \gets F(k_1, l_{j,t})$, \quad $\mathsf{S}_h \gets \mathsf{S}_h \cup \{(w_j, h_j)\}$
    \State $v_j \gets F(k_1, h_j)$, \quad $\mathsf{S}_v \gets \mathsf{S}_v \cup \{(w_j, v_j)\}$

    \State \(\mathsf{wtag}_j \leftarrow g_1^{F_p(v_j, w_j)}\), \, \(\mathsf{ctag}_j \gets H(\mathsf{wtag}_j)\)
    \State \(\mathsf{stag}_j \gets F(k_T, w_j )\), \(\theta_j \gets \mathsf{stag}_j \oplus F(h_j, w_j)\) 
    \State\(C[\mathsf{ctag}_j] \gets \theta_j\),  \(c \gets 1\)
   \For{$i = 1$ \text{ to } $|\mathsf{DB}(w_j)| + \alpha$}
      \State \(\mathsf{x}_i \gets F_p(k_2, ind_i)\), \, \(z_j \gets F_p(k_z, w_j)\) 
      \State \(\mathsf{xtrap}_j \gets F_p(k_x, w_j)\), \, \(y_{i,j} \gets \mathsf{x}_{\mathrm{i}} \cdot z_j^{-1}\) 
      \State \(e_{i,j} \gets \textsf{Sym.Enc}(v_j, ind_i)\), \quad \(l_j \gets H(\mathsf{stag}_j \,\|\, c)\)
      \State \(\mathsf{T}[l_j] \gets (e_{i, j}, y_{i, j})\), \, \(c \gets c + 1\), \quad \(\mathsf{Count}_j \gets c\)
      \State $\mathsf{Count} \gets \mathsf{Count} \cup \{(w_j, \mathsf{Count}_j)\}$

    \State $\mathsf{xtag}_{i,j} \gets g_1^{\,\mathsf{xtrap}_j \cdot \mathsf{x}_i}$; append $\mathsf{xtag}_{i,j}$ to $\mathsf{X}$

    \EndFor
  \EndFor
  \State \Return \((\mathsf{EDB}, \, \, \mathsf{S}_h, \,  \mathsf{S}_v, \mathsf{Count})\)
\end{algorithmic}
\end{algorithm}
\paragraph{\((\mathsf{D}_\mathcal{C},\, \mathsf{keys}_\mathcal{C}) \leftarrow \mathsf{ClientRegistration}_\mathcal{O}\Big( \mathsf{A}_\mathcal{C},\,\{\mathsf{W}_t\}_{t=1}^{|\mathsf{A}_\mathcal{C}|},\,pk_\mathcal{C},\,\mathsf{S}_h,\\,\mathsf{S}_v,\ pk_\mathcal{O},\,sk_\mathcal{O},\,k_x,\,k_z\Big)\).}
Algorithm~\ref{alg:clientRegistration}, run by $\mathcal{O}$ for each client~$\mathcal{C}$, takes as input the client's attribute set $\mathsf{A}_\mathcal{C}$, the associated keyword groups $\{\mathsf{W}_t\}_{t=1}^{|\mathsf{A}_\mathcal{C}|}$, the public key $pk_\mathcal{C}$, and the tables $\mathsf{S}_h = \{(w_j,h_j)\}$ and $\mathsf{S}_v = \{(w_j,v_j)\}$ generated by Algorithm~\ref{alg:edbsetup}.  
These tables contain only the entries corresponding to keywords authorized for ~$\mathcal{C}$.
Each attribute $a_t \in \mathsf{A}_\mathcal{C}$ maps to a subset $\mathsf{W}_t$ of keywords. To construct the authorized keyword set $\overline{W}_\mathcal{C}$, $\mathcal{O}$ aggregates all keywords from these subsets, ensuring uniqueness even in the presence of overlapping attributes.
The values $h_j$ and $v_j$ stored in $\mathsf{S}_h$ and $\mathsf{S}_v$ support token derivation, while the master keys $(k_x,k_z)$ are used during registration.  
$\mathcal{O}$ generates a signature $\sigma_\mathcal{C}$, using his secret key $sk_\mathcal{O}$, attesting the client’s authorization.

Algorithm~\ref{alg:clientRegistration} outputs a dictionary $\mathsf{D}_\mathcal{C}$ and secret key material $\mathsf{keys}_\mathcal{C}$.  
$\mathsf{D}_\mathcal{C}$ contains:
(i) the signature~$\sigma_\mathcal{C}$, and  
(ii) a compact structure $\mathsf{Ctoken}_\mathcal{C}$ that enables $\mathcal{C}$ to derive search tokens for all keywords in~$\overline{W}_\mathcal{C}$.
The signature $\sigma_\mathcal{C}$ is constructed using a cryptographic accumulator, which efficiently represents a set and supports membership verification without revealing individual elements~\cite{acc, acc1}.

For each $\mathcal{C}$, the accumulator aggregates PRF-derived values \(F_p(v_j, w_j)\) for each keyword \(w_j \in \overline{W}_\mathcal{C}\), resulting in
\(\mathsf{Acc}_\mathcal{C} = pk_\mathcal{C}^{\gamma} \in \mathbb{G}_1,
\quad \text{where} \quad
\gamma = \sum_{w_j \in \overline{W}_\mathcal{C}} F_p(v_j, w_j) \in \mathbb{Z}_p^*.\)
$\mathcal{O}$ signs the accumulator as \(\sigma_\mathcal{C} = \mathsf{Acc}_\mathcal{C}^{sk_\mathcal{O}}\).  
During the search phase, $\mathcal{C}$ generates a witness 
$\mathsf{wtag}_j = g_1^{F_p(v_j, w_j)}$ for the least frequent keyword, enabling the server 
$\mathcal{P}$ to verify its authorization and retrieve the corresponding entries in 
$\mathsf{Tset} \, \mathsf{T}$. The remaining keywords in the conjunctive query are then checked using hash-based lookups in $\mathsf{Xset} \, \mathsf{X}$.

The $\mathsf{Ctoken}_\mathcal{C}$ structure, which stores the hashed representations of the client’s authorized keywords, is subsequently used to verify that each remaining keyword belongs to the client’s authorized set. This process ensures query integrity and prevents privilege escalation or unauthorized keyword composition among clients.

Finally, the dictionary \(\mathsf{D}_\mathcal{C}\), indexed by \(pk_\mathcal{C}\), is transmitted to the server \(\mathcal{P}\).   Each registered client, via its attribute set $A_\mathcal{C}$, receives its authorized keyword subset $\overline{W}_\mathcal{C}$ and secret keys $\mathsf{keys}_\mathcal{C}$ for generating valid search tokens.

\begin{algorithm}[h]
\caption{\(\mathsf{ClientRegistration}_{\mathcal{O}}\)}
\label{alg:clientRegistration}
\begin{algorithmic}[1]
  \Statex \textsf{input:} \(\mathsf{A}_\mathcal{C}, \{\mathsf{W}_t\}_{t=1}^{|\mathsf{A}_\mathcal{C}|}, \mathsf{S}_h, \mathsf{S}_v, pk_\mathcal{C}, pk_\mathcal{O}, sk_\mathcal{O}, k_x, k_z\)
  \Statex \textsf{output:} \(\mathsf{D}_\mathcal{C}, \mathsf{keys}_\mathcal{C}\)

  \State \(\mathsf{D}_\mathcal{C} \gets \emptyset\), \quad \(\gamma \gets 0\), \quad \(\mathsf{Ctoken}_\mathcal{C} \gets \emptyset\), \quad \(\overline{W}_\mathcal{C} \gets \emptyset\)

  \For{each \(a_t \in \mathsf{A}_\mathcal{C}\)}
      \For{each \(w_j \in \mathsf{W}_t\)}
          \If{\(w_j \notin \overline{W}_\mathcal{C}\)}
              \State \(\overline{W}_\mathcal{C} \gets \overline{W}_\mathcal{C} \cup \{w_j\}\)
          \EndIf
      \EndFor
  \EndFor

  \For{each \(w_j \in \overline{W}_\mathcal{C}\)}
      \State \(h_j \gets \textsf{parse}(\mathsf{S}_h)\), \, \(v_j \gets \textsf{parse}(\mathsf{S}_v)\)
      \State \(\gamma \gets \gamma + F_p(v_j, w_j)\), \, \(\mathsf{xtrap}_j \gets F_p(k_x, w_j)\)
      \For{each \(ind_i \in \mathsf{DB}(w_j)\)}
          \State \(\mathsf{x}_i \gets F_p(k_2, ind_i)\)
    \State append \(H(g_1^{\mathsf{xtrap}_j \cdot \mathsf{x}_i})\) to \(\mathsf{Ctoken}_\mathcal{C}\)

      \EndFor
  \EndFor

  \State \(\mathsf{Acc}_\mathcal{C} \gets (pk_\mathcal{C})^{\gamma}\), \(\sigma_\mathcal{C} \gets (\mathsf{Acc}_\mathcal{C})^{sk_\mathcal{O}}\)
  \State \(\mathsf{D}_\mathcal{C} \gets \{\sigma_\mathcal{C}, \mathsf{Ctoken}_\mathcal{C}\}\)
  \State \Return \(\mathsf{D}_\mathcal{C}, \mathsf{keys}_\mathcal{C} = \{k_x, k_z, (w_j, h_j, v_j)_{w_j\in\overline{W}_\mathcal{C}}, pk_\mathcal{O}\}\)
\end{algorithmic}
\end{algorithm}
\textbf{$\textsf{Tk}_q \gets \mathsf{TokenGen}_\mathcal{C}(q, \overline{W}_\mathcal{C}, \mathsf{keys}_\mathcal{C}, pk_\mathcal{C}, sk_\mathcal{C})$.}
Algorithm~\ref{alg:TokenGen}, 
run by the authorized client, generates a search token 
$\textsf{Tk}_q$ for a conjunctive query $q=\{w_1,\dots,w_n\}$.
$\mathcal{C}$ uses its key pair $(sk_\mathcal{C},pk_\mathcal{C})$, the set of authorized keywords
$\overline{W}_\mathcal{C}$, and the secret material
\(\mathsf{keys}_\mathcal{C}=\{k_x,k_z,(w_j,h_j,v_j)_{w_j\in\overline{W}_\mathcal{C}},pk_\mathcal{O}\}.\)
To optimize query execution and prevent illegitimate requests, $\mathcal{C}$ selects the least-frequent keyword $w_1$ as pivot and computes the witness \(\mathsf{wtag}_1 = g_1^{F_p(v_1,w_1)} ,\)
from which the server derives $\mathsf{ctag}_1$ and $\mathsf{stag}_1$ to locate matching entries in $\mathsf{Tset}$.
The remaining keywords are aggregated such that:
\(A_{\mathsf{xtoken}} = g_1^{\sum_{w_j \in \overline{W}_\mathcal{C} \setminus \{w_1\}} F_p(v_j,w_j)} ,\)
and for each $w_j$ ($j\ge2$) the client computes an extended token as follows: \(\mathsf{xtoken}[j] = g_1^{F_p(k_x,w_j)\cdot F_p(k_z,w_1)} ,\) which allows the server to check that any document containing the pivot also contains $w_j$. Since the client provides $\mathsf{wtag}_1$ and $A_{\mathsf{xtoken}}$, the server can verify that the witness $\mathsf{wtag}_1$ corresponds to an authorized keyword, analogous to accumulator-based membership verification.
The final token is
\(\textsf{Tk}_q =
\{pk_\mathcal{C},\mathsf{wtag}_1,\ A_{\mathsf{xtoken}},
pk_\mathcal{O}^{sk_\mathcal{C}}, F(h_1,w_1),\mathsf{xtoken}[j]\}_{j=2}^{n} \}.\)
It is transmitted to the cloud server, which processes the query using 
Algorithm~\ref{alg:search}.
\begin{algorithm}[h]
\caption{\(\mathsf{TokenGen}_{\mathcal{C}}\), ($w_1$ is the least frequent keyword)}
\label{alg:TokenGen}
\begin{algorithmic}[1]
    \Statex \textsf{input:}  \(q = \{w_1, \dots, w_n\} \subseteq \overline{W}_\mathcal{C}\), \(\mathsf{keys}_{c} = \{k_x, k_z, (w_j, h_j, v_j)_{w_j \in \overline{W}_\mathcal{C}}, pk_\mathcal{O}\} \) ,    \; \(pk_\mathcal{C}\), \, \(sk_\mathcal{C}\)
     \Statex \textsf{output:} \( \textsf{Tk}_{q} \)
     \State\( \mathsf{wtag}_1 \gets g_1^{F_p(v_1, \, w_1)} \), \; \(F\Bigl(h_1, \, w_1\Bigr)\) \;  
    \State\(A_{\mathsf{xtoken}} \; \gets \; g_1^{\sum_{w_j \in \overline{W}_\mathcal{C} \setminus \{w_1\}} F_p\bigl(v_j, \, w_j\bigr)}\)
        \For{\( j \gets 2 \) \textsf{to} \( n \)}
            \State \( \textsf{xtoken}[j] \gets g_1^{\; {F_p(k_{x}, \, w_j)} \cdot F_p(k_z, \, w_1 )} \) 
        \EndFor
\State
\(
\mathsf{Tk}_q \gets
\begin{aligned}[t]
\Bigl\{ & pk_{\mathcal C},\,
\mathsf{wtag}_1,\,
A_{\mathsf{xtoken}}, \\
& pk_{\mathcal O}^{sk_{\mathcal C}},\,
F(h_1,w_1),\,
\{\mathsf{xtoken}[j]\}_{j=2}^{n}
\Bigr\}
\end{aligned}
\)
    \State \Return \( \textsf{Tk}_q \)
\end{algorithmic}
\end{algorithm}\\

\textbf{\(\textsf{R} \leftarrow \mathsf{Search}_\mathcal{P}(\mathsf{EDB}, \mathsf{D}_{c}, \textsf{Tk}_q)\).}  
Algorithm~\ref{alg:search}, run by $\mathcal{P}$, verifies the client’s token $\textsf{Tk}_q$ and returns the corresponding encrypted identifiers.  
It takes  as input the encrypted database $\mathsf{EDB}$, the client’s dictionary $\mathsf{D}_\mathcal{C}$, and $\textsf{Tk}_q$.  
First, $\mathcal{P}$ verifies the client’s authentication by checking the signature:
\(\hat{e}(\sigma_\mathcal{C}, g_2) \stackrel{?}{=} \hat{e}(A_{\mathsf{xtoken}} \cdot \mathsf{wtag}_1, pk_{\mathcal{O}}^{sk_\mathcal{C}}),\) if verification fails, the query is rejected. 
Then, $\mathcal{P}$ computes
\(\mathsf{stag}_1 = \theta_1 \oplus F(h_1, w_1)\), with \(\theta_1 \gets \mathsf{Cset}[\mathsf{ctag}_1]\), and retrieves the matching entries from $\mathsf{Tset}$.
For each pair \((e_{i,j}, y_{i,j})\) and each $\mathsf{xtoken}[j]$, it checks:
\(
\forall j \ge 2, \quad \mathsf{xtoken}[j]^{y_{i,j}} \in X \quad\text{and}\quad H(\mathsf{xtoken}[j]^{y_{i,j}}) \in \mathsf{Ctoken}_\mathcal{C}.\)  
If both hold, $e_{i,j}$ is added to the result set $\textsf{R}$, which the client decrypts using $v_1$.  
This ensures that even colluding clients cannot escalate privileges, as each step enforces strict cryptographic verification.

\begin{algorithm}[h]
\caption{\(\mathsf{Search}_{\mathcal{P}}\)}
\label{alg:search}
\begin{algorithmic}[1]
 \Statex \textsf{inputs:} \(\mathsf{EDB}\), \(\mathsf{D}_\mathcal{C}\), \(\textsf{Tk}_{q}\)
 \Statex \textsf{Output:} \(\textsf{R}\)
\State \(\textsf{R} \gets \emptyset\)
\If{\(\hat{e}\Bigl(\sigma_\mathcal{C}, g_2\Bigr) \neq
 \hat{e}\Bigl(A_{\mathsf{xtoken}} \cdot \mathsf{wtag}_1, pk_{\mathcal{O}}^{sk_\mathcal{C}}\Bigr)\)}
 \State \textsf{Reject:} Access denied
    \State \Return \(\emptyset\)
\EndIf
\State \(\mathsf{ctag}_1 \gets H(\mathsf{wtag}_1)\)
\State \(\theta_1 \gets C[\mathsf{ctag}_1]\), \quad \(\mathsf{stag}_1 \gets \theta_1 \oplus F\Bigl(h_1, w_1\Bigr)\), \(c \gets 1\)
\While{\textsf{true}}
    \State \(l_1 \gets H\Bigl(\mathsf{stag}_1 \| c\Bigr)\)
    \If{\(T[l_1] = \text{null}\)}
        \State \textsf{break}
    \EndIf
    \State \((e_{i, j}, y_{i, j}) \gets T[l_1]\)
    \For{\(j \gets 2\) \textsf{to} \(n\)}
        \If{\(\mathsf{xtoken}[j]^{y_{i, j}} \in X\) \textsf{and} \(H\Bigl(\mathsf{xtoken}[j]^{y_{i, j}}\Bigr) \in \mathsf{Ctoken}_\mathcal{C}\)}
            \State \(\textsf{R} \gets \textsf{R} \cup \{e_{i, j}\}\)
        \EndIf
    \EndFor
    \State \(c \gets c + 1\)
\EndWhile
\State \Return \(\textsf{R}\)
\end{algorithmic}
\end{algorithm}
\textbf{Update Phase.}
This protocol run by $\mathcal{O}$ and $\mathcal{P}$, enables secure and efficient dynamic updates. 
During index construction, $\mathcal{O}$ inserted $\alpha$ dummy documents per keyword in $\mathsf{Tset}$, 
allowing up to $\alpha$ updates without rebuilding the index. 
Let $op \in \{\mathsf{add}, \mathsf{del}\}$ denotes an update operation for keyword $w_j$ and document $ind_i$. 
The procedure depends on the operation type:

\begin{itemize}
    \item \(op = \mathsf{add}\):  
    $\mathcal{O}$ sets the update counter \(c \gets |\mathsf{DB}(w_j)| - (\alpha - 1)\), corresponding to the first dummy slot reserved for future updates, and computes the index label  
    \(l_j \gets H(\mathsf{stag}_j \Vert c)\), where \(\mathsf{stag}_j = F(k_T, w_j)\).
    For subsequent updates of the same keyword \(w_j\), \(c\) is decremented by 1 to use the next reserved dummy slot.
    Than, $\mathcal{O}$ derives the cryptographic values:   
    \(x_i \gets F_p(k_2, ind_i)\), \(xtrap_j \gets F_p(k_x, w_j)\), and \(z_j \gets F_p(k_z, w_j)\),  
    generates the pair \(\mathsf{Tset}\)
   \(y_{i,j} \gets x_i \cdot z_j^{-1}\), \(e_{i,j} \gets \textsf{Sym.Enc}(v_j, ind_i)\),  
    and computes the corresponding \(\mathsf{xtag}\)  
    \(\mathsf{xtag}_{i,j} \gets g_1^{\,xtrap_j \cdot x_i}\).
    $\mathcal{O}$ sends the tuple  
    \((\mathsf{op}, l_j, (y_{i,j}, e_{i,j}), \mathsf{xtag}_{i,j})\)  
    to $\mathcal{P}$, which replaces the first dummy entry in \(\mathsf{Tset}[l_j]\) with \((y_{i,j}, e_{i,j})\) and inserts \(\mathsf{xtag}_{i,j}\) into \(\mathsf{Xset}\).
    \item \(op = \mathsf{del}\): $\mathcal{O}$ retrieves the location  \(l_j \gets H(\mathsf{stag}_j \Vert c)\) and computes the corresponding \(\mathsf{xtag}_{i, j}\) for the document to be deleted. Then, $\mathcal{O}$ sends \((l_j, \mathsf{xtag})\) to \(\mathcal{P}\), who removes the pair \((e_{i, j},y_{i, j})\) at 
    position \(l_j\) in \(\mathsf{Tset}\) and removes the associated \(\mathsf{xtag}_{i, j}\) 
    from \(\mathsf{Xset}\). 
\end{itemize}
This mechanism enables constant-time updates by using pre-inserted dummy entries, while preserving the encrypted index structure.

\textbf{Revocation.}
To revoke a client $\mathcal{C}$, $\mathcal{O}$ sends $pk_\mathcal{C}$ to $\mathcal{P}$, which deletes the corresponding authorization dictionary $\mathsf{D}_\mathcal{C}$. 
For granular keyword revocation, $\mathcal{O}$ computes \(\gamma' = \gamma - F_p(v_j, w_j)\)
and updates the signature as
\(\sigma'_\mathcal{C} = pk_\mathcal{C}^{\gamma'\cdot^sk_\mathcal{O}}.\)
The corresponding entries are removed from $\mathsf{Ctoken}_\mathcal{C}$. 
Subsequent search tokens from this client fail verification. Revocation does not require re-encryption of the database nor rebuilding the index.

\section{Security Analysis}
\label{sec: security-analysis}
This section discusses the security properties of \textsf{MASSE} based on the cryptographic assumptions presented in Definitions~\ref{app:crypto2}, \ref{app:crypto3} and \ref{app:crypto5} in Appendix~\ref{app:crypto}.

\subsection{Indistinguishability of Data and Queries\label{proof:1}}
To analyze the security of \textsf{MASSE} against an adversarial server, we formalize the leakage function \(\mathcal{L}\) of our scheme. Part of the leakage is inherited from the \textsf{OXT} protocol ~\cite{Cash14}, while additional leakage arises from the specific access control mechanism introduced in \textsf{MASSE}.
Let \(q = (s, x, \mathsf{CD}_c)\) be a sequence of \(Q\) non-adaptive 2-conjunctive queries, where each query is written as:
\(q(t) = \bigl(s(t),\, x(t),\, \mathsf{CD}_c(t)\bigr), \quad t \in [Q].\)
Here, \(s(t)\) denotes the least frequent keyword (s-term), \(x(t)\) is the second keyword (x-term), and \(\mathsf{CD}_c(t)\) is the client dictionary \(\mathcal{C}\) used in a query \(t\). For clarity, we focus on the 2-conjunctive case, but the formulation naturally generalizes to \(n\)-conjunctive queries. Let \(m\) be the number of distinct keywords in the database.
The leakage function \(\mathcal{L} = (N, \mathsf{N}_c, \mathsf{N}_{c,j}, \hat{s}, \mathsf{SP}, \mathsf{RP}, \mathsf{IP})\), which is exposed to the server, consists of the following components:
\begin{itemize}
    \item \textsf{Total keyword occurrences \(N\):} the total number of keywords appearances in the database, \(N = \sum_{j=1}^D |W_j|\).
    
    \item \textsf{Client dictionary size \(\mathsf{N}_c\):} for the \(c\)-th client, the number of keywords in their dictionary \(\mathsf{CD}_c\), is defined as:
    \(\mathsf{N}_c = |\mathsf{CD}_c|.\)
    
    \item \textsf{Per-keyword access leakage \(\mathsf{N}_{c,j}\):} a binary value indicating whether the keyword \(W_j\) belongs to the dictionary of the \(c\)-th client:
    \(\mathsf{N}_{c,j} =
    \begin{cases}
    1 & \text{if } W_j \in \mathsf{CD}_c \\
    0 & \text{otherwise}
    \end{cases}\)
    
    \item \textsf{Equality pattern \(\hat{s} \in (m)^Q\):} a vector indicating which queries share the same s-term. It assigns a unique integer to each distinct s-term in the order of appearance.
    
    \item \textsf{Size pattern \(\mathsf{SP} \in (d)^Q\):} the number of documents that match the first keyword in each query, \(\mathsf{SP}(t) = |\mathsf{DB}(s(t))|\).
    
    \item \textsf{Result pattern \(\mathsf{RP}\):} the set of documents that match both keywords in each query, \(\mathsf{RP}(t) = \mathsf{DB}(s(t)) \cap \mathsf{DB}(x(t))\).
    
     \item \textsf{IP} conditional intersection pattern.  For all $i \neq j$, with $t_i, t_j \in [Q]$ and $\beta_1, \beta_2 \in [n]$, where $n$ denotes the number of $x$-terms in a query, the pattern is defined as:
\noindent\makebox[\columnwidth][c]{%
\resizebox{0.92\columnwidth}{!}{$
\displaystyle
\mathsf{IP}[t_i,t_j;\beta_1,\beta_2]=
\begin{cases}
DB(s[t_i])\cap DB(s[t_j]), & x[t_i,\beta_1]=x[t_j,\beta_2],\\
\emptyset, & \text{otherwise.}
\end{cases}
$}}
   \item \(\mathsf{CD}_c(t)\) reveals the client identity for the \(t\)-th query through its public key \(pk_\mathcal{C}\) and associated authorization data

\end{itemize}
\begin{theorem}
\label{theo:indis}
Let $\mathcal{L}$ denote the leakage function defined above. Under the assumption of DDH in $\mathbb{G}_1$, with secure PRFs $F, F_p$, and an IND-CPA secure symmetric encryption scheme $\Sigma = (\mathsf{Sym.Enc}, \mathsf{Dec})$, the \textsf{MASSE} scheme is $\mathcal{L}$-semantically secure against non-adaptive adversaries. That is, for every PPT adversary $\mathcal{A}$, there exists a PPT simulator $\mathcal{S}$ such that
\[\left| \Pr[\mathsf{Real}_{\mathcal{A}}(\lambda) = 1] - \Pr[\mathsf{Ideal}_{\mathcal{A}, \mathcal{S}}(\lambda) = 1] \right| \leq \mathsf{negl}(\lambda)\]
\end{theorem}

Before presenting the proof of this theorem, we first define the real and ideal worlds for the adversary, namely the server.
\begin{itemize}
    \item \textsf{Real Experiment} \label{world:reel} $\mathsf{Real}_{\mathcal{A}}(\lambda)$:  
    The adversary $\mathcal{A}$ chooses a database $\mathsf{DB}$ together with the associated attribute and keyword relations 
    $(\mathsf{A}, \{\mathsf{W}_t\}_{t=1}^{|\mathsf{A}|}, \{\mathsf{A}_j\}_{j=1}^{D})$. Here, $\mathsf{A}$ denotes the set of attributes, 
    while $\{\mathsf{W}_t\}$ and $\{\mathsf{A}_j\}$ define the mappings between attributes and keywords. 
    Next, $\mathcal{A}$ selects a non-adaptive sequence of $Q$ queries $q = (q(1), \ldots, q(Q))$. 
    The challenger then runs $\mathsf{Setup}$, $\mathsf{KeyGen}$, and $\mathsf{EDBSetup}$ to obtain 
    $(\mathsf{pp}, \mathsf{Keys}, \mathsf{EDB})$, and executes $\mathsf{ClientRegistration}$ for each client. 
    For each query $q(t)$, the challenger runs $\mathsf{TokenGen}$ and $\mathsf{Search}$, which output a pair $(\mathsf{Res}_t, \mathsf{ResInds}_t)$ consisting of encrypted results and the corresponding matching document identifiers. All observable outputs are collected into a transcript $T$.  
    Finally, $\mathcal{A}$ receives $(\mathsf{EDB}, T)$ and outputs a bit $b \in \{0,1\}$.

    \item \textsf{Ideal Experiment} \label{world:ideal} $\mathsf{Ideal}_{\mathcal{A}, \mathcal{S}}(\lambda)$:  
    The adversary $\mathcal{A}$ chooses a database $\mathsf{DB}$ with the corresponding attribute and keyword relations 
    $(\mathsf{A}, \{\mathsf{W}_t\}_{t=1}^{|\mathsf{A}|}, \{\mathsf{A}_j\}_{j=1}^{D})$. Here, $\mathsf{A}$ is the set of attributes, 
    and $\{\mathsf{W}_t\}$ and $\{\mathsf{A}_j\}$ represent the mappings between attributes and keywords. 
    It then selects a non-adaptive sequence of $Q$ queries $q = (q(1), \ldots, q(Q))$. 
    Given only the leakage $\mathcal{L}(\mathsf{DB}, q)$, a simulator $\mathcal{S}$ generates a simulated encrypted database 
    $\mathsf{EDB}$ and a simulated transcript $T$. These are provided to $\mathcal{A}$, who finally outputs a bit 
    $b \in \{0,1\}$ as its guess.
\end{itemize}

\begin{proof} To prove Theorem \ref{theo:indis}, we define a sequence of security games \(\textsf{G}_0, \dots, \textsf{G}_n\), each slightly modifying the previous. We argue that each pair \((\textsf{G}_i, \textsf{G}_{i+1})\) is computationally indistinguishable, thus, an adversary \(\mathcal{A}\) cannot distinguish the real from the ideal execution, proving security.
The hash function \(H\) is modeled as a random oracle maintaining a table \(\mathcal{M}\) mapping inputs to outputs. On a query $x \in \{0,1\}^*$, if $x \in \mathcal{M}$ the oracle returns $\mathcal{M}[x]$; otherwise, it samples
$y \xleftarrow{\$} \{0,1\}^\lambda$, stores $(x,y)$ in $\mathcal{M}$ and returns $y$.

\paragraph*{Game $\mathsf{G}_0$.}
Game~$\mathsf{G}_0$ models the real execution with minor simplifications.
On input $(\mathsf{DB}, q, K, \mathcal A, \{\mathsf W_t\}_{t=1}^{|\mathcal A|}, \{\mathsf A_j\}_{j=1}^D)$,
$\mathcal C$ initializes $\mathsf{EDB}$, constructs $\mathsf{XSet}$, and generates 
authorization dictionaries $\mathsf{D}_{\mathcal C}$ and secret keys $\mathsf{keys}_u$.
The view of $\mathcal{A}$ consists of the public parameters, the server-side structures
$\mathsf{Tset}$, $\mathsf{Xset}$, and $\mathsf{Cset}$, and the output of all its queries.
For each query $t \in [Q]$, the encrypted results $\mathsf{Res}$ and indices
$\mathsf{ResInds}$ are computed as
$\mathsf{ResInds} = \mathsf{DB}(s[t]) \cap \mathsf{DB}(x[t])$.
The only difference from real execution is the absence of false positives.
Assuming $F_p$ is a secure PRF, thus
\[
\bigl|\Pr[\mathsf{G}_0 = 1] - \Pr[\textsf{Real}^{\textsf{MASSE}}(\lambda) = 1]\bigr|
\leq \mathsf{negl}(\lambda)
\]
\paragraph*{Game $\mathsf{G}_1$.}  
In $\mathsf{G}_1$, all pseudo-random functions $F$ and $F_p$ used in the construction are replaced by truly random functions. 
Specifically, inputs previously evaluated with $F_p$ now use uniform random elements $f_2, f_j, f_x, f_z \in \mathbb{Z}_p^*$, and inputs previously evaluated with $F$ use uniform random functions over $\{0,1\}^\lambda$.  
All other operations are identical to $\mathsf{G}_0$.  
By a standard sequence-of-games argument, there exist adversaries $\mathcal{B}_{1,1}$ and $\mathcal{B}_{1,2}$ such that
\[
\big|\Pr[\mathsf{G}_0 = 1] - \Pr[\mathsf{G}_1 = 1]\big|
    \leq 5 \cdot \text{Adv}^{F}_{\mathcal{B}_{1,1}}(\lambda)
    + 4 \cdot \text{Adv}^{F_p}_{\mathcal{B}_{1,2}}(\lambda)\] 

\paragraph*{Game $\mathsf{G}_2$.}  
In $\mathsf{G}_2$, the adversary cannot distinguish a fake ciphertext from a real one, assuming the symmetric encryption scheme $\Sigma$ is IND-CPA secure.  
The $\mathsf{Tset}$ structure is modified so that each ciphertext $e_{i,j}$ encrypts a fixed string of zeros instead of the real index.  
All other operations remain as in $\mathsf{G}_1$.  
If an adversary $\mathcal{A}$ can distinguish $\mathsf{G}_2$ from $\mathsf{G}_1$, one can construct an adversary $\mathcal{B}_2$ that breaks the IND-CPA security of $\Sigma$. Thus
\[
\left|\Pr[\mathsf{G}_2 = 1] - \Pr[\mathsf{G}_1 = 1]\right|
    \leq m \cdot \text{Adv}^{\text{IND-CPA}}_{\Sigma,\mathcal{B}_2}(\lambda),\]
where $m$ is the number of symmetric keys used in $\mathsf{Tset}$.
\paragraph*{Game $\mathsf{G}_3$.} 
Game $\mathsf{G}_3$ transitions from pseudorandom functions to precomputed random arrays. During initialization, the game samples $H[w], Z[w] \xleftarrow{\$} \mathbb{Z}_p^*$ for each $w \in W$ and $I[ind] \xleftarrow{\$} \mathbb{Z}_p^*$ for each $ind$. These values are used to construct the protocol components as follows: 
(i) for each $w \in W$ and $ind \in DB(w)$, $\mathsf{xtag} = g_1^{H[w] \cdot I[ind]}$ is added to $\mathsf{Xset}$, and $\mathsf{Tset}$ stores $y = I[ind] \cdot Z[w]^{-1}$, (ii) for each query, the search tokens are set to $\mathsf{xtoken}[j] = g_1^{Z[w_1] \cdot H[w_j]}$, for  $j \ge 2$. 
Since the arrays $H, Z,$ and $I$ contain uniform elements and the algebraic mapping is preserved, the adversary's view is identical to $\mathsf{G}_2$ $\Pr[\mathsf{G}_3 = 1] = \Pr[\mathsf{G}_2 = 1]$
\paragraph*{Game $\mathsf{G}_4$.} 
In Game $\mathsf{G}_4$, we replace the $\mathsf{xtag}$ values in $\mathsf{Xset}$ with uniform random elements from $\mathbb{G}_1$. For each $(w,ind)$ with $w \in W$ and $ind \in DB(w)$, $\mathsf{xtag}$ is sampled as $g_1^r$ for $r \xleftarrow{\$} \mathbb{Z}_p^*$.
This transition is indistinguishable from $\mathsf{G}_3$ under the DDH assumption. Let $(g_1^a, g_1^b, g_1^c)$ be a DDH challenge. For each pair $(w,ind)$, the reduction algorithm $\mathcal{B}_4$ embeds the instance setting $H[w] = a$ and $I[ind] = b$, while sampling all $Z[w]$ uniformly at random. It programs $\mathsf{xtag} = g_1^c$ and constructs $\mathsf{xtoken}[i] = (g_1^a)^{Z[w_1]}$ together with the blinding value $y = b \cdot Z[w_1]^{-1}$ consistently with $\mathsf{G}_3$. During the search, the adversary computes \((\mathsf{xtoken}[i])^y = ((g_1^a)^{Z[w_1]})^{b \cdot Z[w_1]^{-1}} = g_1^{ab}.\)
In $\mathsf{G}_3$, this value is compared with $\mathsf{xtag} = g_1^c$ where $c = ab$, whereas in $\mathsf{G}_4$ it is compared with $g_1^r$. Hence, any advantage in distinguishing the two games yields a DDH distinguisher
\[
\left| \Pr[\mathsf{G}_4 = 1] - \Pr[\mathsf{G}_3 = 1] \right|
\leq \mathrm{Adv}^{\mathrm{DDH}}_{\mathbb{G}_1,\mathcal{B}_4}(\lambda)\]
\paragraph*{Game $\mathsf{G}_5$.}
Game~$\mathsf{G}_5$ modifies the random-oracle and array-access behavior. The simulator is based only on the allowed leakage.  
The oracle $\mathcal{M}$ is honestly evaluated only for keywords that appear as $s$-terms of queries. Otherwise, it outputs uniform values.  Array accesses are never performed in isolation. The $\mathsf{Tset}$ value $y$ requires simultaneous access to $I[ind]$ and $Z[w]$. This occurs in two cases.  
(i) For some $t \in [Q]$ and $\beta_1 \in [n]$ ($n$ denotes the number of $x$-terms in a query), the keyword $w$ is the $s$-term of query $t$.  

The file $ind$ contains both $w$ and the associated x-term $x_{\beta_1}(t)$.  
(ii) Two distinct queries $t_1, t_2$ have different $s$-terms that share an x-term.  For some $\beta_1, \beta_2 \in [n]$, the file $ind$ contains both $s$-terms.  
In all other cases, $y$ is sampled uniformly in $\mathbb{G}_1$.  
Similarly, an $\mathsf{xtag}$ in $\mathsf{Xset}$ requires simultaneous access to $H[w]$ and $I[ind]$.  
This happens only when $w$ is an x-term of query $t$ and $ind$ contains the corresponding $s$-term. Otherwise, it is sampled uniformly. These modifications ensure that the simulator remains consistent with the permitted leakage. They also preserve the game’s distribution.
Thus
\[\left|\Pr[\mathsf{G}_5 = 1] - \Pr[\mathsf{G}_4 = 1]\right| \leq \mathrm{negl}(\lambda)\]

\paragraph*{Simulator.}
For our analysis, we construct a simulator $\mathcal{S}$ that, given the leakage function $\mathcal{L}(DB,s,x) = (N, N_u, N_{u,j}, \overline{s}, SP, RP, IP)$, outputs a simulated encrypted database $\mathsf{EDB}$ and transcript $T$.  

By combining the transitions from $\mathsf{G}_0$ to $\mathsf{G}_5$,  we obtain that the joint distribution produced in Game~$\mathsf{G}_5$ is computationally indistinguishable from the real execution. It therefore suffices to show that the output $(\mathsf{EDB}, T)$ generated by $\mathcal{S}$ is identically distributed to that of Game~$\mathsf{G}_5$.

More precisely, upon receiving $\mathcal{L}(DB,s,x)$, the simulator generates
$\mathsf{EDB} = (\mathsf{Tset}, \mathsf{Cset}, \mathsf{Xset})$ together with a transcript $T$ whose distribution matches that of Game~$\mathsf{G}_5$.

The simulator uses the conditional intersection leakage $IP$ to derive a restricted equality patterns of the query $x$, denoted $\hat{x} \in [m]^{Q \times n}$.  
Intuitively, $\hat{x}$ specifies, for each query, which x-terms must be treated as equivalent.  

Formally, $\mathcal{S}$ defines a relation $\equiv$ over $[Q,n]$, where two positions $[t_1,\beta_1]$ and $[t_2,\beta_2]$ are related whenever
$\mathsf{IP}[t_1,t_2;\beta_1,\beta_2] \neq \emptyset$  

The transitive closure of this relation produces an equivalence relation, and each resulting equivalence class receives a unique identifier, $\hat{x}[t,\beta_1]$ is then defined as the index of the class containing $[t,\beta_2]$.
As in the \textsf{OXT} scheme, $\hat{x}$ fully captures the equality pattern of the query $x$.  In particular,

\(
\hat{x}[t_1,\beta_1] = \hat{x}[t_2,\beta_2]
\Rightarrow x[t_1,\beta_1] = x[t_2,\beta_2], \\
(x[t_1,\beta_1] = x[t_2,\beta_2]) \wedge
(DB(s[t_1]) \cap DB(s[t_2]) \neq \emptyset)
\Rightarrow \hat{x}[t_1,\beta_1] = \hat{x}[t_2,\beta_2].\)

In the simulation, the arrays $H$, $I$ and $Z$ are sampled uniformly at random, consistent with the allowed leakage. 

As a result, the adversary observes values that are identically distributed to those  in the real construction and obtains no information beyond the prescribed leakage.
Unlike Game~$\mathsf{G}_5$, where keyword and document identifiers are used explicitly, the simulator treats all keywords as indices in $[Q]$ and derives all arrays solely from leakage.
The simulated encrypted database
$\mathsf{EDB}$ and transcript $T$ are generated as follows.

For $\mathsf{Cset}$, the simulator creates an entry per keyword-related element and fills all remaining positions with uniform independent strings. Since Game~$\mathsf{G}_5$ follows the same distribution, thte resulting $\mathsf{Cset}$ is identically distributed.

For $\mathsf{Tset}$, the simulator follows the construction of
Game~$\mathsf{G}_5$, replacing identifiers not appearing in the equivalence relation with dummy symbols whose associated values are sampled uniformly from
$\mathbb{Z}_p$. This preserves the distribution of $\mathsf{Tset}$.

For $\mathsf{Xset}$, the simulator replaces each pair
$(\mathrm{ind}, x[t,\beta])$ with $(\mathrm{ind}, \hat{x}[t,\beta])$, using the same indices derived from $\mathsf{RP}$ or $\mathsf{IP}$. Two entries are equal if and only if their corresponding $\hat{x}$ values are equal, since in Game~$\mathsf{G}_5$ all group elements are sampled uniformly at random subject only to the equality constraints induced by the leakage. 
Consequently, the simulated $\mathsf{Xset}$ maintains the same equivalence structure and has the same distribution as in Game~$\mathsf{G}_5$.

The transcript $T$ produced by the simulator has the same distribution as the real execution, meaning that the difference in the adversary's probability of success is negligible
\[
\left| \Pr[\mathsf{Real}_{\mathcal{A}}(\lambda) = 1] - \Pr[\mathsf{Ideal}_{\mathcal{A}, \mathcal{S}}(\lambda) = 1] \right| \leq \mathsf{negl}(\lambda)
\]
\end{proof}

\subsection{Forward and Backward Privacy}
We prove that \textsf{MASSE} provides forward privacy and Type-II backward privacy against a semi-honest PPT cloud server.
\begin{corollary}[Forward and Type-II Backward Privacy of \textsf{MASSE}]
Assuming that $F$ and $F_p$ are secure pseudorandom functions and that the Decisional Diffie-Hellman (DDH) assumption holds in $\mathbb{G}_1$, the \textsf{MASSE} scheme achieves forward privacy and Type-II backward privacy.
\end{corollary}

\begin{proof}
We prove both properties using a common real/ideal indistinguishability framework with a sequence-of-games argument.

Before presenting the proof of this corollary, we define the leakage functions considered in the analysis.
Let $Q$ be a sequence of queries consisting of search queries $(t,w)$ issued with the timestamp $t$ for the keyword $w$, and update queries $(t, op,(w,ind))$ issued with the timestamp $t$, where $\mathsf{op} \in \{\mathsf{add},\mathsf{del}\}$.

\begin{itemize}

    \item The function $\mathsf{TimeDB}(w)$ returns the set of document identifiers associated with keyword $w$ that have not been deleted, together with their insertion timestamps:
    \(\mathsf{TimeDB}(w)
    =\{(t,ind) |
    (t,\mathsf{add},(w,ind)) \in Q \ \wedge\ \nexists\, t' :
    (t',\mathsf{del},(w,ind)) \in Q\}.\)
    
    \item For a conjunctive query $q = (w_1 \wedge \cdots \wedge w_n)$, we extend this definition as
    \(\mathsf{TimeDB}(q)
    =\{( t_1,\ldots,t_n\rangle,ind) \ |\
    \forall i \in [n],\
    (t_i,\mathsf{add},(w_i,ind)) \in Q \ \wedge\ (w_i,ind) \text{ is not deleted}.\)
    
    \item The function $\mathsf{Upd}(w)$ returns the set of timestamps in which the keyword $w$ is involved in an update operation:
    \[\mathsf{Upd}(w)
    =\bigl\{t \mid \exists\, \mathsf{op},ind :
    (t,\mathsf{op},(w,ind)) \in Q \bigr\}\]
    
    \item For a conjunctive query
    $q = (w_1 \wedge \cdots \wedge w_n)$, where $w_1$ denotes the s-term, the
    aggregate update leakage is defined as:
    \(\mathsf{Upd}(q)
    =\mathsf{Upd}(w_1)\ \cup\
    \bigcup_{i=2}^{n} \mathsf{Upd}(w_1,w_i)\)
    
    For instance, for a two-keyword query $q=(w_1 \wedge w_2)$, the corresponding x-term update leakage is
    \(\mathsf{Upd}(w_1,w_2)
    =\{(t_1,t_2) \ |\
    \exists\, ind,\mathsf{op} :
    (t_1,\mathsf{op},(w_1,ind)) \in Q \ \\ \land\ (t_2,\mathsf{op},(w_2,ind)) \in Q\}\)
    
    \item Let $\mathsf{DelHist}(w)$ denote the deletion history of keyword $w$,
    recording the insertion and deletion timestamps of documents containing $w$:
    \(\mathsf{DelHist}(w)
    =\{(t_{\mathsf{add}},t_{\mathsf{del}}) \ |\
    \exists\, ind :
    (t_{\mathsf{add}},\mathsf{add},(w,ind)) \in Q
    \ \land\
    (t_{\mathsf{del}},\mathsf{del},(w,ind)) \in Q\}\)
    
\end{itemize}

\paragraph*{Forward Privacy.}

We formalize forward privacy by comparing the server views generated by the update protocol in a real $\mathsf{Real}_{\mathcal{A}}(\lambda)$ execution and in an ideal simulation $\mathsf{Ideal}_{\mathcal{A},\mathcal{S}}(\lambda)$.
\begin{itemize}
    \item In $\mathsf{Real}_{\mathcal{A}}(\lambda)$ experiment, the challenger $\mathcal{C}$ initializes the system and constructs the encrypted database with $\alpha$ dummy slots per keyword in $\mathsf{Tset}$. $\mathcal{A}$ adaptively selects two update sequences of equal length, differing only in the identifiers and keywords of new documents. $\mathcal{C}$ executes one sequence honestly using the \textsf{Update} algorithm, overwriting the corresponding dummy slots at positions $l_j = H(\mathsf{stag}_j \| c)$, and returns the resulting server view. $\mathcal{A}$ outputs a bit $b \in \{0,1\}$ indicating which sequence was executed.

    \item In  $\mathsf{Ideal}_{\mathcal{A},\mathcal{S}}(\lambda)$ experiment, $\mathcal{A}$ issues the same update sequences. A simulator $\mathcal{S}$, given only the leakage functions $\mathsf{TimeDB}(w)$, $\mathsf{TimeDB}(q)$, $\mathsf{Upd}(w)$ and  $\mathsf{Upd}(q)$  (revealing the timing and number of updates but not the content), produces a simulated server view that is given to $\mathcal{A}$. The adversary outputs a bit $b \in \{0,1\}$ defined as in the real experiment.
\end{itemize}
We prove forward privacy for \textsf{MASSE} via a sequence of games.

\paragraph*{Game $\mathsf{G}_0$.}  
$\mathsf{G}_0$ corresponds to the real execution of \textsf{MASSE}, where $\mathcal{A}$ interacts with the system normally and observes search tokens, update tags, and encrypted indices.

\paragraph*{Game $\mathsf{G}_1$.}  
All outputs of the PRFs $F$ and $F_p$ are replaced with uniformly independent random values. In particular, the search tokens $\mathsf{stag}_j$, the update tags $\mathsf{xtag}_{i,j}$, and the auxiliary values are sampled uniformly at random.  
By PRF security, there exist adversaries $\mathcal{B}_{1}$ and $\mathcal{B}_{2}$ such that
\[
\big|\Pr[\mathsf{G}_0=1] - \Pr[\mathsf{G}_1)=1]\big|
\leq \mathrm{Adv}^{\mathsf{PRF}}_{F, \mathcal{B}_{1}}(\lambda)
+ 3\cdot\mathrm{Adv}^{\mathsf{PRF}}_{F_p,  \mathcal{B}_{2}}(\lambda)\]
\paragraph*{Game $\mathsf{G}_2$.}  
Each ciphertext $e_{i,j}$ in $\mathsf{Tset}$ is replaced by an encryption of a fixed zero string. By the IND-CPA security of $\Sigma$, there exists an adversary $\mathcal{B}_3$ such that
\[\big|\Pr[\mathsf{G}_1)=1] - \Pr[\mathsf{G}_2)=1]\big|
\leq \mathrm{Adv}^{\mathsf{IND\text{-}CPA}}_{\Sigma, \mathcal{B}_3}(\lambda)\]

\paragraph*{Game $\mathsf{G}_3$.}  
The structured exponentiation used to compute $\mathsf{xtag}_{i,j}$ is replaced by values drawn from a simulated DDH oracle in $\mathbb{G}_1$. Any distinguisher $\mathcal{A}$ between $\mathsf{G}_2$ and $\mathsf{G}_3$ can be used to construct a DDH adversary $\mathcal{B}_4$
\[
\big|\Pr[\mathsf{G}_2=1] - \Pr[\mathsf{G}_3=1]\big|
\leq \mathrm{Adv}^{\mathsf{DDH}}_{\mathbb{G}_1, \mathcal{B}_4}(\lambda)\]

\paragraph*{Simulation.}  
Using only the leakage functions $\mathsf{Upd}(w)$, $\mathsf{Upd}(q)$ and $\mathsf{TimeDB}(q)$, a simulator $\mathcal{S}$ can generate a server view identically distributed to $\mathsf{G}_3$. Therefore, the forward privacy advantage of any  $\mathcal{A} $ is \(\mathsf{Adv}^{\mathrm{Forward}}_{\textsf{MASSE}}(\mathcal{A})\) and is bounded by
\[\mathrm{Adv}^{\mathsf{PRF}}_{F, \mathcal{B}_{1}}(\lambda)
+
3\cdot \mathrm{Adv}^{\mathsf{PRF}}_{F_p, \mathcal{B}_{2}}(\lambda)
+
\mathrm{Adv}^{\mathsf{IND\text{-}CPA}}_{\Sigma, \mathcal{B}_3}(\lambda)
+
\mathrm{Adv}^{\mathsf{DDH}}_{\mathbb{G}_1, \mathcal{B}_4}(\lambda)\]

\paragraph*{Type-II Backward Privacy.}We establish Type-II backward privacy for \textsf{MASSE} by comparing the server views obtained after document deletions in the real execution and in an ideal simulation.
\begin{itemize}
    \item In $\mathsf{Real}_{\mathcal{A}}(\lambda)$ experiment, the challenger $\mathcal{C}$ initializes the system and grants the adversary adaptive access to valid document additions. $\mathcal{A}$ selects a previously added document $d^*$ and requests its deletion. $\mathcal{C}$ performs the deletion by overwriting the corresponding entry in $\mathsf{Tset}$ and removing the associated tag from $\mathsf{Xset}$, then returns the resulting server view to $\mathcal{A}$, who outputs a bit $b \in \{0,1\}$.

    \item In $\mathsf{Ideal}_{\mathcal{A},\mathcal{S}}(\lambda)$ experiment, a simulator $\mathcal{S}$, given only the leakage functions $\mathsf{TimeDB}(w)$, $\mathsf{Upd}(w)$ and $\mathsf{DelHist}(w)$, generates a simulated server view consistent with the leakage and provides it to $\mathcal{A}$, who outputs a bit $b \in \{0,1\}$.
\end{itemize}

We prove indistinguishability between the real and ideal deletion experiments
via a sequence of games.

\paragraph*{Game $\mathsf{G}_0$.}
Game $\mathsf{G}_0$ is identical to the real execution of the deletion protocol,
except that all evaluations of the PRFs $F$ and $F_p$ are replaced with uniformly
random functions.
Any distinguisher between the real experiment and $\mathsf{G}_0$ yields PRF
adversaries $\mathcal{B}_{0,1}$ and $\mathcal{B}_{0,2}$ such that
\[\big|\Pr[\mathsf{Real}_{\mathcal{A}}(\lambda)=1]
- \Pr[\mathsf{G}_0=1]\big|
\leq
\mathrm{Adv}^{\mathsf{PRF}}_{F, \mathcal{B}_{0,1}}(\lambda)
+
3\mathrm{Adv}^{\mathsf{PRF}}_{F_p, \mathcal{B}_{0,2}}(\lambda)\]

\paragraph*{Game $\mathsf{G}_1$.}
In $\mathsf{G}_1$, each ciphertext $e_{i,j}$ stored in $\mathsf{Tset}$ encrypts a
fixed zero string instead of the corresponding document identifier.
By the IND-CPA security of the symmetric encryption scheme $\Sigma$, there exists
an adversary $\mathcal{B}_1$ such that
\[
\big|\Pr[\mathsf{G}_0=1]
- \Pr[\mathsf{G}_1=1]\big|
\leq
\mathrm{Adv}^{\mathsf{IND\text{-}CPA}}_{\Sigma, \mathcal{B}_1}(\lambda)\]

\paragraph*{Game $\mathsf{G}_2$.}
In this game, the structured exponentiation used to generate update tags
$\mathsf{xtag}_{i,j}$ is replaced with values obtained from a simulated
Decisional Diffie--Hellman oracle in $\mathbb{G}_1$.
Any adversary distinguishing $\mathsf{G}_1$ from $\mathsf{G}_2$ can be transformed
into a DDH adversary $\mathcal{B}_2$, yielding
\[|\Pr[\mathsf{G}_1=1]
- \Pr[\mathsf{G}_2=1]|
\leq
\mathrm{Adv}^{\mathsf{DDH}}_{\mathbb{G}_1, \mathcal{B}_2}(\lambda)\]
\paragraph*{Simulation (Ideal Game).}
Using only leakage functions
$\mathsf{Upd}(w)$ and $\mathsf{DelHist}(w)$, a simulator $\mathcal{S}$ can generate random overwrites and tag removals whose
distribution is identical to that of $\mathsf{G}_2$ and independent of the target document $ind^*$.

Therefore,
\(
\mathsf{Adv}^{\mathrm{Backward\text{-}II}}_{\textsf{MASSE}}(\mathcal{A})
\leq
\mathrm{Adv}^{\mathsf{PRF}}_{F, \mathcal{B}_{0,1}}(\lambda)
+
3\cdot\mathrm{Adv}^{\mathsf{PRF}}_{F_p, \mathcal{B}_{0,2}}(\lambda)
+
\mathrm{Adv}^{\mathsf{IND\text{-}CPA}}_{\Sigma, \mathcal{B}_1}(\lambda)
+
\mathrm{Adv}^{\mathsf{DDH}}_{\mathbb{G}_1, \mathcal{B}_2}(\lambda)\)

Under the stated assumptions, \textsf{MASSE} achieves forward privacy and Type-II backward privacy.
\end{proof}

\subsection{Unforgeability against Malicious Clients\label{proof:2}}
We show that \textsf{MASSE} is unforgeable against malicious clients, meaning that no client can generate a valid search token for a keyword it is not authorized to query for any keyword $w^* \notin \overline{W}_\mathcal{C}$.

\begin{theorem}
\label{theo:proofUnf}
Under the Random Oracle Model and assuming that $F_p$ is a secure pseudorandom function, the \textsf{MASSE} scheme is unforgeable against malicious clients. Let $\overline{W}_\mathcal{C}$ denote the set of keywords authorized for the client $\mathcal{C}$.  
For any probabilistic polynomial-time adversary $\mathcal{A}$, the probability of generating a valid search token $\mathsf{Tk}_{w^*}$ for any $w^* \notin \overline{W}_\mathcal{C}$ is negligible:
 \(\mathrm{Adv}_{\textsf{MASSE}}^{\mathrm{Forge}}(\mathcal{A}) \leq \mathsf{negl}(\lambda)\)
\end{theorem}
The proof of Theorem \ref{theo:proofUnf} is given in Appendix \ref{app:proofunf}.

\section{Performance Analysis}
\label{sec: Dis}
This section discusses the performance and practical considerations of \textsf{MASSE}, including implementation details and evaluation results.

\subsection{Theoretical Analysis}
\label{sec:theorical}
As shown in Table~\ref{app:computation_cost}, the $\mathsf{EDBSetup}$ cost for \textsf{OXT} depends solely on the number of keyword--identifier pairs $N$. 
\textsf{MASSE} also scales linearly in $N$, with an additional $D$ term due to the $\mathsf{Cset}$ structure, which captures client-specific keyword permissions. 
On the other hand, \textsf{SEAC} and \textsf{DMSSE} result in higher $\mathsf{EDBSetup}$ costs: in \textsf{SEAC}, setup depends on the number of registered clients, whereas in \textsf{DMSSE}, it scales with the number of keywords each client is authorized to access, leading to data duplication across clients (cf. Table~\ref{app:comm_cost}) and requiring additional bilinear pairing operations.
Thus, although \textsf{MASSE} is more computationally demanding during the enrollment $\mathsf{ClientRegistration}$ process, it achieves a setup that is independent of the total number of clients, as reflected in Table~\ref{app:comm_cost} for both encrypted database size.
Regarding $\mathsf{ClientRegistration}$ in Table~\ref{app:computation_cost}, \textsf{SEAC} and \textsf{DMSSE} are more efficient than \textsf{MASSE} since their registration cost depends either on the number of registered client (\textsf{SEAC}) or on the maximum number of authorized clients per keyword (\textsf{DMSSE}). 
In \textsf{MASSE}, registration scales with the number of keyword--identifier document pairs the client is authorized to search. 
However, this overhead is incurred only once per client and ensures that subsequent operations, including search, communication, and storage, remain independent of the total number of clients. 
This comes at the cost of storing one dictionary for each client. 
This design allows for an efficient revocation of a client without affecting others, providing superior scalability for large multi-client environments.

$\mathsf{TokenGen}$ depends only on the number of queried keywords, whereas in \textsf{OXT} and \textsf{SEAC},
it additionally depends on the number of documents matching the least frequent keyword due to the generation of $\mathsf{xtokens}$.
Consequently, \textsf{MASSE} achieves lower token generation cost than \textsf{OXT} and \textsf{SEAC}. It also achieves smaller token sizes than \textsf{OXT} and \textsf{SEAC}. \textsf{MASSE} remains comparable to \textsf{DMSSE} in terms of computation, communication, and storage.

For $\mathsf{Search}$, \textsf{MASSE} outperforms \textsf{OXT} and \textsf{SEAC} by reducing the number of required exponentiations, and improves upon \textsf{DMSSE}, which relies on more expensive bilinear pairing operations. Finally, \textsf{MASSE} supports constant-time revocation without client
interaction; \textsf{SEAC} however, requires all remaining authorized clients to update their witnesses.

In conclusion, except for the $\mathsf{ClientRegistration}$ operation, \textsf{MASSE} computations outperform \textsf{DMSSE} and \textsf{SEAC}, and even \textsf{OXT} in terms of computation time.

\begin{table}[h]
    \caption{Comparative computation costs for SSE schemes}
    \centering
    \renewcommand{\arraystretch}{0.9}
    \resizebox{\columnwidth}{!}{%
    \begin{tabular}{lcccc}
        \toprule
        \textsf{Scheme} & \textsf{OXT}~\cite{Cash14} & \textsf{DMSSE}~\cite{DU20} & \textsf{SEAC}~\cite{Hu2025SEAC} & \textbf{\textsf{MASSE}} \\
        \midrule
        \textsf{EDBSetup} & $N \cdot \textsf{exp}$ 
                          & $N \cdot \textsf{exp} + N \cdot U\cdot \textsf{bp}$ 
                          & $N \cdot C \cdot \textsf{exp}$ 
                          & \textbf{$(N + D) \cdot \textsf{exp}$} \\
        \textsf{ClientReg.} & -- 
                          & $U \cdot \textsf{exp}$ 
                          & $C \cdot \textsf{exp}$ 
                          & \textbf{{$C \cdot(K\cdot P +2 ) \cdot \textsf{exp}$}} \\
        \textsf{TokenGen} & $(n-1) \cdot P \cdot \textsf{exp}$ 
                          & $n\cdot \textsf{exp}$ 
                          & $n \cdot P \cdot \textsf{exp}$ 
                          & \textbf{{$(n+1) \cdot \textsf{exp}$}} \\
        \textsf{Search} & $(n-1) \cdot P \cdot \textsf{exp}$ 
                        & $P \cdot \textsf{bp}$ 
                        & $n \cdot P \cdot \textsf{exp}$ 
                        & \textbf{{$P \cdot \textsf{exp} + 2 \cdot \textsf{bp}$}} \\
        \textsf{Update} & -- 
                        & -- 
                        & $D \cdot C \cdot \textsf{exp}$ 
                        & \textbf{{$D \cdot \alpha \cdot \textsf{exp}$}} \\
        \textsf{Revocation} & -- 
                            & -- 
                            & $T \cdot \textsf{exp}$ 
                            & \textbf{{$1$}} \\
        \bottomrule
    \end{tabular}%
    }
    \label{app:computation_cost} 
    \vspace{2mm}

  \footnotesize{
    \textsf{exp}: exponentiations; \textsf{bp}: bilinear pairings; 
    $D$: total number of keywords; $N$: number of keyword--document identifier pairs; $n$: number of queried keywords; $C$: total number of registered clients; $U$: maximum number of authorized clients per keyword; $T$: number of non-revoked clients; $P = |\mathrm{DB}(w)|$: maximum number of documents per keyword; $K = |\overline{W}_\mathcal{C}|$: maximum number of keywords per client; $\alpha$: number of dummy entries per keyword, added to hide access patterns.}
\end{table}

\begin{table}[h]
    \caption{Comparative communication and storage overheads}
    \centering
    \resizebox{\columnwidth}{!}{%
    \begin{tabular}{lcc}
        \toprule
        \textsf{Scheme} & \textsf{EDB size + (AC Size*)} & \textsf{Token Size} \\
        \midrule
        \textsf{OXT}~\cite{Cash14} 
        & $(N \cdot (\lambda + |\mathbb{Z}_p^*| + |\mathbb{G}_1|))$
        & $\lambda + (n-1) \cdot P \cdot |\mathbb{G}_1|$ \\
        
        \textsf{DMSSE}~\cite{DU20} 
        & $(N \cdot (\lambda + |\mathbb{G}_2| + U \cdot |\mathbb{G}_T|)
           + (U \cdot D\cdot \lambda)$
        & $2\lambda +(n-1) \cdot |\mathbb{G}_1|$ \\
        
        \textsf{SEAC}~\cite{Hu2025SEAC} 
        & $(N \cdot C \cdot (2\cdot P\cdot D\lambda + |\mathbb{Z}_p^*| + |\mathbb{G}_1|)) 
           + ( |\mathbb{Z}_N^*|)$
        & $n \cdot P \cdot |\mathbb{G}_1| + P \cdot |\mathbb{Z}_N^*|$ \\
        
        \textbf{\textsf{MASSE}} 
        & \textbf{$(N \cdot (\lambda + |\mathbb{Z}_p^*| + |\mathbb{G}_1|) 
           + D \cdot \lambda) 
           + (C \cdot K \cdot P \cdot \lambda+ C\cdot| \mathbb{G}_1|)$}
        & \textbf{$\lambda + n \cdot |\mathbb{G}_1| + |\mathbb{G}_2|$} \\
        \bottomrule
    \end{tabular}%
    }
    
    \vspace{2mm}
    \label{app:comm_cost}
  \footnotesize{
  *: AC size refers to the size of the incorporated access control elements.

  $|\mathbb{Z}_p^*|$, $|\mathbb{Z}_N^*|$, $|\mathbb{G}_1|$, $|\mathbb{G}_2|$, $|\mathbb{G}_T|$: group element sizes;  
  $\lambda$: symmetric ciphertext size.
  For other notations, refer to Table~\ref{app:computation_cost}.
  }
\end{table}

\begin{figure}[h]
  \centering
  \begin{subfigure}[h]{0.70\columnwidth}
    \centering
    \includegraphics[width=\linewidth]{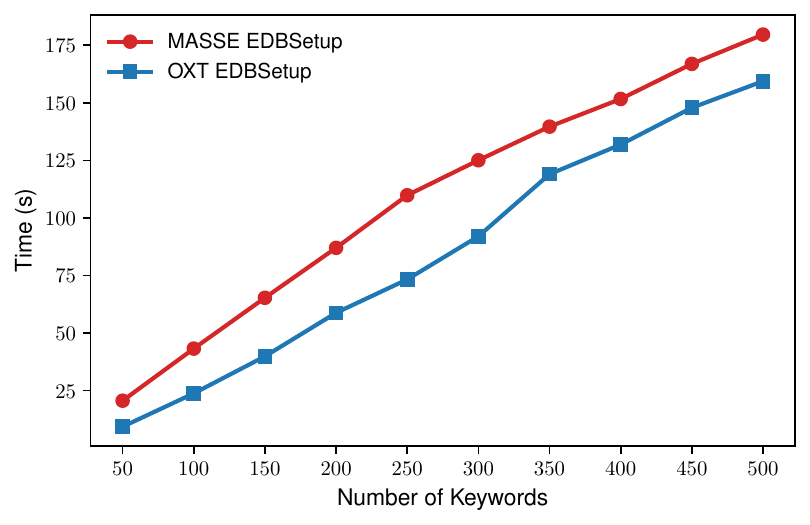}
    \caption{$\mathsf{EDBSetup}$ time for 50--500 keywords.}
    \label{fig:edbgen}
  \end{subfigure}
  \begin{subfigure}[h]{0.70\columnwidth}
    \centering
    \includegraphics[width=\linewidth]{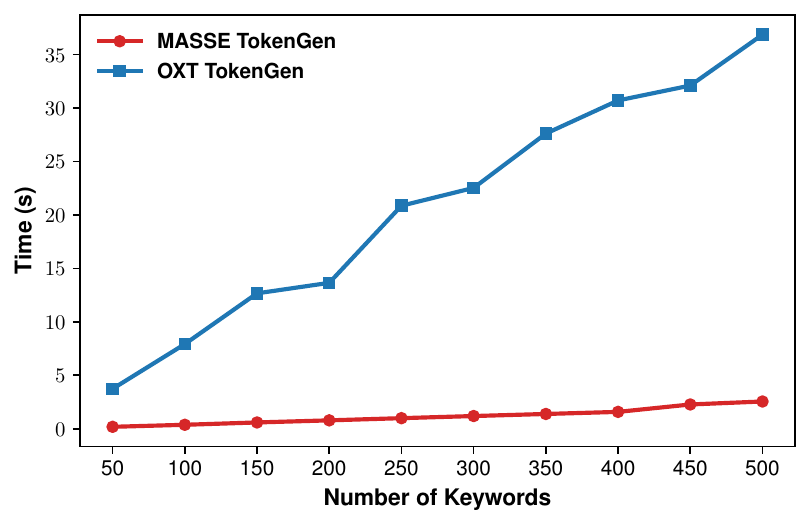}
    \caption{$\mathsf{TokenGen}$ time for 50--500 keywords.}
    \label{fig:tokengen}
  \end{subfigure}
  \begin{subfigure}[h]{0.70\columnwidth}
    \centering
    \includegraphics[width=\linewidth]{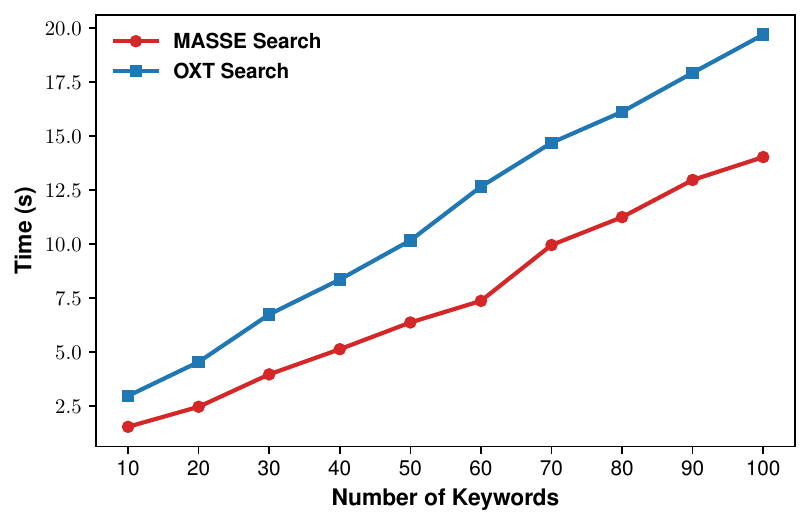}
    \caption{$\mathsf{Search}$ time for 10--100 keywords.}
    \label{fig:search}
  \end{subfigure}
  \caption{Experimental computation time --  \textsf{MASSE} and \textsf{OXT}}
  \label{fig:performance}
\end{figure}
\subsection{Implementation and Settings}
Our experiments were conducted on a host machine running \textit{Windows 11}, equipped with an \textit{Intel}(R) Core(TM) i7-8665U CPU  1.90\,GHz and 32\,GB of RAM. 
The experiments were run inside an \textit{Ubuntu} 22.04 LTS virtual machine configured with 4 virtual CPUs, 22\,GB of RAM, and 50\,GB of allocated storage. 
Our programs are implemented in \textit{C} and compiled using gcc. 
We rely on the \textit{PBC} library for group operations and bilinear pairings, and on \textit{OpenSSL} for symmetric encryption and hash functions. 
We select symmetric elliptic curve groups $(\mathbb{G}_1, \mathbb{G}_2, \mathbb{G}_T)$ defined over the curve $y^2 = x^3 + x$. The security parameter $\lambda$ is set to 256 bits. For cryptographic operations, we use a single secure hash function based on \textsf{HMAC-SHA256} for all hashing and pseudo-random function computations, and \textsf{AES} in CTR mode for encryption.
\subsection{Experimental Results}
The efficiency of \textsf{MASSE} is measured against that of \textsf{OXT}~\cite{Cash14} with respect to the performance of the three processes: the encrypted database generation ($\mathsf{EDBSetup}$), the token generation ($\mathsf{TokenGen}$), and the search ($\mathsf{Search}$). 
Each keyword appears in 200 documents. All reported times are averages over 10 runs.

Fig.~\ref{fig:edbgen} reports the execution time of the
$\mathsf{EDBSetup}$ algorithm. As predicted in Table~\ref{app:computation_cost}, the $\mathsf{EDBSetup}$ cost in both schemes scales linearly with the number of keywords. For a dataset with 500 keywords, 200 documents, and 500 attributes per keyword, \textsf{MASSE} requires 178~s, compared to 149~s for \textsf{OXT}. This moderate overhead mainly arises from constructing the $\mathsf{Cset}$ structure, which involves additional exponentiations. In contrast, attribute processing is based on lightweight PRF operations.

Fig.~\ref{fig:tokengen} presents search token generation times. Unlike \textsf{OXT}, which cost depends both on the number of keywords queried and the frequency of the document of the least frequent keyword, \textsf{MASSE} depends only on the number of keywords queried. Consequently, for a query of 500 keywords, token generation takes about 2~s in \textsf{MASSE}, versus 40~s in \textsf{OXT},
confirming the scalability of \textsf{MASSE} with respect to query size.

Fig.~\ref{fig:search} compares search performance over an encrypted database with 200 keywords, each appearing in 150 documents. For queries ranging from 10 to 100 keywords, \textsf{MASSE} consistently outperforms \textsf{OXT}, even though \textsf{OXT} targets only a single-client setting. In particular, for 100-keyword queries, \textsf{MASSE} completes the search in about 14~s, while \textsf{OXT} requires approximately 19~s. This gap is due to the large number of $\mathsf{xtokens}$ verified in \textsf{OXT}, whereas \textsf{MASSE} significantly reduces the number of membership checks, leading to lower computational overhead during search.
The experimental results are consistent with the theoretical analysis of Section ~\ref{sec:theorical}, confirming that \textsf{MASSE} scales efficiently with both the size of the dataset and the complexity of the query.
\section{Conclusion}
\label{sec:con}
We presented \textsf{MASSE}, a practical dynamic multi-client SSE construction with incorporated attribute-based access control. After a one-time setup by the data owner, authorized clients can perform conjunctive keyword searches without further interaction. 
Access control is enforced through token generation and prevents unauthorized queries, even in the event of collusion between clients. The server verifies search tokens without learning the queried keywords or client attributes. The security of \textsf{MASSE} is formally proven. 
Our performance evaluation in computation shows that \textsf{MASSE} outperforms competing solutions and even \textsf{OXT}, and can scale to large encrypted databases.
\textsf{MASSE} is well-suited to deployment in multi-tenant setting, such as cloud storage and collaborative platforms, where controlled and privacy-preserving searches over encrypted data are required. Its design strikes a balance between expressiveness, security, and efficiency in dynamic environments.

{\footnotesize
\bibliographystyle{abbrvnat}
\bibliography{MASSE}
}
\appendix
\section{Cryptographic Preliminaries}
\label{app:crypto}

This section introduces the cryptographic tools and assumptions used in \textsf{MASSE}. We define bilinear pairings, relevant complexity assumptions, and symmetric encryption with IND-CPA security, which are essential for our construction and analysis.

\subsection{Bilinear Pairings}
\label{pair: paring}

\begin{definition}[Type III Bilinear Pairing]
Let $\mathbb{G}_1$ and $\mathbb{G}_2$ be cyclic groups of prime order $p$. A bilinear pairing is a map $\hat{e}: \mathbb{G}_1 \times \mathbb{G}_2 \to \mathbb{G}_T$ satisfying:
\begin{itemize}
    \item \textsf{Bilinearity:} $\hat{e}(g_1^a, g_2^b) = \hat{e}(g_1, g_2)^{ab}$ for all $g_1 \in \mathbb{G}_1$, $g_2 \in \mathbb{G}_2$, and $a,b \in \mathbb{Z}_p^*$.
    \item \textsf{Non-degeneracy:} $\hat{e}(g_1, g_2) \neq 1$ for some $g_1, g_2$.
    \item \textsf{Efficient Computability:} $\hat{e}(g_1, g_2)$ can be computed efficiently.
\end{itemize}
Type III means no efficient isomorphism exists between $\mathbb{G}_1$ and $\mathbb{G}_2$~\cite{pair}. A parameter generator outputs $(p, \mathbb{G}_1, \mathbb{G}_2, \mathbb{G}_T, \hat{e}, g_1, g_2)$.
\end{definition}

\subsection{Complexity Assumptions}
\label{app:crypto1}

\begin{definition}[Decisional Diffie-Hellman (DDH)]
\label{app:crypto2}
Let $g$ be a generator of a cyclic group $\mathbb{G}$ of prime order $p$. DDH states that no PPT adversary can distinguish:
\[(g, g^a, g^b, g^{ab}) \quad \text{from} \quad (g, g^a, g^b, g^c)\]
for random $a,b,c \in \mathbb{Z}_p^*$. Formally,
\[
\begin{aligned}
\mathrm{Adv}^{\mathrm{DDH}}_{A,\mathbb{G}}(\lambda)
&= \Big|
   \Pr[A(g,g^a,g^b,g^{ab})=1] \\
&\quad - \Pr[A(g,g^a,g^b,g^c)=1]
   \Big| \\
&\le \mathrm{negl}(\lambda)
\end{aligned}
\]
\end{definition}



\begin{definition}(Pseudorandom function).
\label{app:crypto3}
Let $X, Y$ be finite sets, and let $F: \{0,1\}^\lambda \times X \to Y$ be a function. $F$ is a pseudo-random function (PRF) if no probabilistic polynomial-time (PPT) adversary can distinguish between access to $F(k, \cdot)$ for a random key $k$ and access to a truly random function $f \in \mathrm{Fun}(X,Y)$. Formally, the advantage of any adversary $A$ is bounded by a negligible function of the security parameter $\lambda$
\[
\begin{aligned}
\mathrm{Adv}^{\mathrm{PRF}}_{F,A}(\lambda)
&= \left| \Pr[A^{F(k,\cdot)}(1^\lambda)=1]
      - \Pr[A^{f(\cdot)}(1^\lambda)=1] \right| \\
&\le \mathrm{negl}(\lambda)
\end{aligned}
\]
\end{definition}

\subsection{Symmetric Encryption}
\label{app:crypto4}

\begin{definition}(IND-CPA security).
\label{app:crypto5}
Let $\mathsf{SE} = (\mathsf{KeyGen}, \mathsf{Enc}, \mathsf{Dec})$ be a symmetric encryption scheme where 
$\mathsf{Enc}(k,m) \rightarrow ct$ and $\mathsf{Dec}(k,ct) \rightarrow m$, and which correctness is proven as follows: 

$\mathrm{Dec}(k, \mathrm{Enc}(k,m)) = m$.\\
$\mathsf{SE}$ is IND-CPA secure if for any PPT adversary $\mathcal{A}$,  
\[
\begin{aligned}
\mathrm{Adv}^{\mathrm{IND\text{-}CPA}}_{\mathsf{SE},\mathcal{A}}(\lambda) 
  &= \Big| \Pr\big[\mathcal{A}(\mathsf{Enc}_k(m_0))=1\big] \\
  &\quad - \Pr\big[\mathcal{A}(\mathsf{Enc}_k(m_1))=1\big] \Big| \\
  &\leq \mathrm{negl}(\lambda)
\end{aligned}
\]

where $k \leftarrow \mathsf{KeyGen}(1^\lambda)$ and $m_0,m_1$ are equal-length messages.
\end{definition}

\section{Correctness of Client Verification Queries}
\label{appendix:verification}

Let $D_\mathcal{C} = (\sigma_\mathcal{C}, \mathsf{Ctoken}_\mathcal{C})$ be the verification data associated with a registered client.
A client query consists of the tuple:
\[(\mathsf{wtag}_1, A_{\mathsf{xtoken}}, pk_{\mathcal{O}}^{sk_\mathcal{C}}, \, F(h_1, w_1),  \{\mathsf{xtoken}[j]\}_{j=2}^n)\]

\paragraph{Pairing-based verification.}
Let $\overline{W}_\mathcal{C}$ be the set of authorized keywords of the client.
The server checks the following pairing equation:
\[
\begin{aligned}
\hat{e}\bigl(A_{\mathsf{xtoken}} \cdot \mathsf{wtag}_1, pk_{\mathcal{O}}^{sk_\mathcal{C}}\bigr)
&= \hat{e}\Bigl(\prod_{w_j \in \overline{W}_\mathcal{C}} g_1^{F_p(v_j, w_j)}, g_2^{sk_{\mathcal{O}} \cdot sk_\mathcal{C}}\Bigr) \\
&= \hat{e}\Bigl(g_1^{\sum_{w_j \in \overline{W}_\mathcal{C}} F_p(v_j, w_j) \cdot sk_\mathcal{C}}, g_2^{sk_{\mathcal{O}}}\Bigr) \\
&= \hat{e}\Bigl(pk_\mathcal{C}^{\sum_{w_j \in \overline{W}_\mathcal{C}} F_p(v_j, w_j) \cdot sk_{\mathcal{O}}}, g_2\Bigr) \\
&= \hat{e}(\sigma_\mathcal{C}, g_2)
\end{aligned}
\]
This equality holds for any registered client with valid keys and confirms the correctness of the main token.

\paragraph{Verification using $\mathsf{Ctoken}$}
For each additional keyword position $j = 2,\dots,n$, the server computes
\(\mathsf{xtoken}[j]^{y_{i,j}} = g_1^{\mathsf{xtrap}_j \cdot x_i}\), verifies that
\(H\bigl(\mathsf{xtoken}[j]^{y_{i,j}}\bigr) \in \mathsf{Ctoken}_\mathcal{C}.\)
Since $\mathsf{Ctoken}_\mathcal{C}$ contains only hash values associated with authorized keywords, this check guaranties that each auxiliary token is consistent with the client’s authorization set.

Together, pairing verification and $\mathsf{Ctoken}$ membership tests ensure that any query generated by a registered client is correctly verified and accepted if and only if it corresponds to authorized keywords.

\section{Proof of Token's Unforgeability}
\label{app:proofunf}
\begin{proof}
We prove the unforgeability of \textsf{MASSE} stated in Theorem~\ref{theo:proofUnf} via a game-based experiment between a challenger $\mathcal{C}$ and a probabilistic polynomial-time adversary $\mathcal{A}$ modeling a malicious but legitimately registered client.

\textbf{Setup Phase.}  
$\mathcal{C}$ runs $\mathsf{Setup}_{\mathcal{O}}(1^\lambda)$ to generate the public parameters $\mathsf{pp}$ and executes $\mathsf{KeyGen}_{\mathcal{O}}(\mathsf{pp})$ to obtain the data owner's key pair $(sk_{\mathcal{O}}, pk_{\mathcal{O}})$ and the master keys $K$. Public values $(\mathsf{pp}, pk_{\mathcal{O}}, K)$ are given to $\mathcal{A}$, while $sk_{\mathcal{O}}$ remains secret.

The adversary $\mathcal{A}$ generates a key pair $(sk_{\mathcal{A}}, pk_{\mathcal{A}}) \gets \mathsf{KeyGen}_{\mathcal{A}}(\mathsf{pp})$ and submits $pk_{\mathcal{A}}$ together with an attribute set $\mathsf{A}_{\mathcal{A}}$ of its choice. Using $(pk_{\mathcal{A}}, \mathsf{A}_{\mathcal{A}})$, $\mathcal{C}$ executes $\mathsf{ClientRegistration}_{\mathcal{C}}$ to compute the authorized keyword set $\overline{W}_{\mathcal{A}}$,  the accumulator as follows:
\[
\mathsf{Acc}_{\mathcal{A}} = pk_{\mathcal{A}}^\gamma, \quad \gamma = \sum_{w \in \overline{W}_{\mathcal{A}}} F_p(v_w, w),
\]
as well as the signature
\(\sigma_{\mathcal{A}} = \mathsf{Acc}_{\mathcal{A}}^{sk_{\mathcal{O}}}, \)
and the verification structure $\mathsf{Ctoken}_{\mathcal{A}}$.

Finally, $\mathcal{C}$ only sends $\mathcal{A}$ the client-side keys: 
\[
\mathsf{keys}_{\mathcal{A}} = \{k_x, k_z, (w, h_w, v_w)_{w \in \overline{W}_{\mathcal{A}}}, pk_{\mathcal{O}} \},
\]
$\mathcal{A}$ is then able to generate search tokens for authorized keywords, while the signature and verification structures remain fully controlled by $\mathcal{C}$.

\textbf{Challenge Phase.}
$\mathcal{C}$ selects one or several keywords $w^* \notin \overline{W}_{\mathcal{A}}$
(possibly authorized for another client) and provides $\mathcal{A}$ with the corresponding public values $(h_{w^*}, v_{w^*})$.

This models a stronger adversary that may either obtain this public information
through collusion with other registered clients or receive it explicitly as
additional knowledge from $\mathcal{C}$. In both cases, neither signature
$\sigma_{\mathcal{A}}$ nor verification structure $\mathsf{Ctoken}_{\mathcal{A}}$
is updated during this phase.

Using its full knowledge of the system and the challenge keywords,
$\mathcal{A}$ outputs a conjunctive search token $\mathsf{Tk}_{q^*}$
such that $q^*$ contains at least one keyword $w^* \notin \overline{W}_{\mathcal{A}}$.
$\mathcal{A}$ wins if $\mathsf{Tk}_{q^*}$ is accepted by the verification
procedure $\mathsf{Search}$.

We show that the success probability is negligible.
We distinguish two cases.

\textbf{Case 1}: $w^*$ is the pivot keyword. In this case,
\(
A_{\mathsf{xtoken}} = g_1^{\gamma},
\mathsf{wtag}_1^* = g_1^{F_p(v_{w^*}, w^*)}.\)
The pairing verification requires:
\[\hat{e}(\sigma_{\mathcal{A}}, g_2)
=
\hat{e}(g_1^{\gamma + F_p(v_{w^*}, w^*)}, pk_{\mathcal{O}}^{sk_{\mathcal{A}}})
\]
Since
\(
\hat{e}(\sigma_{\mathcal{A}}, g_2)
=
\hat{e}(g_1^{\gamma}, pk_{\mathcal{O}}^{sk_{\mathcal{A}}}),
\)
the equality holds if and only if
\(F_p(v_{w^*}, w^*) \equiv 0 \pmod{p}\)

Because $F_p$ is a pseudo-random function with output space $\mathbb{Z}_p^*$,
all values used to construct elements of $\mathsf{Cset}$ are of the form
$g_1^{F_p(v_j,w_j)}$, where the exponent belongs to $\mathbb{Z}_p^*$.

From the adversary’s point of view, the value $F_p(v_{w^*}, w^*)$ is
computationally indistinguishable from a uniformly random element of
$\mathbb{Z}_p^*$. Hence, $\mathcal{A}$ cannot force or predict any specific
exponent value with non-negligible probability without breaking the PRF
security of $F_p$.

The event $F_p(v_{w^*}, w^*) \equiv 0 \pmod{p}$ can only occur if $F_p$
fails to sample from $\mathbb{Z}_p^*$, which happens with probability exactly
$1/p$, and is therefore negligible in the security parameter $\lambda$.

Even if this negligible event occurs, we obtain $\mathsf{wtag}_1^* = g_1^0 = 1$.
$\mathcal{C}$ then computes $H(\mathsf{wtag}_1^*) = H(1)$ in order to access $\mathsf{Cset}$. However, all valid entries of $\mathsf{Cset}$ correspond to hashes of group elements of the form $H(g_1^{F_p(v_j,w_j)})$ with $F_p(v_j,w_j) \in \mathbb{Z}_p^*$.
Therefore, no entry is indexed by $H(1)$ unless a collision occurs in the hash function $H$, which happens with negligible probability under the collision resistance assumption. As a result, the search procedure does not return document identifiers.

Thus, the success probability in this case is at most $1/p$.

\textbf{Case 2}: $w^*$ is a non-pivot keyword.

For any position $j \geq 2$,
\(\mathsf{xtoken}^*[j] = g_1^{F_p(k_x, w^*) \cdot F_p(k_z, w_1)}.\)
Acceptance requires that, for some retrieved entry,
\[H(\mathsf{xtoken}^*[j]^{y_{i,j}}) \in \mathsf{Ctoken}_{\mathcal{A}}.
\]
By construction, $\mathsf{Ctoken}_{\mathcal{A}}$ contains only hash values
associated with authorized keywords in $\overline{W}_{\mathcal{A}}$.
Therefore, $\mathcal{A}$ must either:
(i) find a collision in $H$, or
(ii) compute a preimage of an existing hash value.
Modeling $H$ as a random oracle with $\lambda$-bit output,
the probability of either event is at most $2^{-\lambda}$.

By the union bound, the overall success probability satisfies
\[\mathrm{Adv}_{\textsf{MASSE}}^{\mathrm{Forge}}(\mathcal{A}) \leq \mathsf{negl}(\lambda)\]
\end{proof}

\end{document}